\documentclass[12pt]{article}

\usepackage[sort&compress]{natbib}
\usepackage{mathrsfs}
\usepackage{amsfonts}
\usepackage{tabularray}
\usepackage{amssymb}
\usepackage{dsfont,footmisc}
\usepackage[all]{xy}

\usepackage{graphicx}
\usepackage{amsmath}  
\usepackage{rotating}
\usepackage{appendix}
\usepackage{longtable}
\usepackage{booktabs}
\usepackage{multirow}

\usepackage{color}
\usepackage{subcaption}
\usepackage{comment}
\usepackage{tabularx}
\usepackage[font=small]{caption}%
\usepackage[colorlinks,linkcolor=red,anchorcolor=blue,citecolor=blue]%
{hyperref}

\newtheorem{theorem}{Theorem}[section] 
\newtheorem{lemma}{Lemma}[section] 
\newtheorem{remark}{Remark}[section] 
\newtheorem{corollary}{Corollary}[section] 
\newtheorem{example}{Example}[section]
\numberwithin{equation}{section}
\newtheorem{definition}{Definition}[section]
\allowdisplaybreaks[4]

\newenvironment{proof}[1][Proof]{\noindent \textbf{#1.} }{\  \rule{0.5em}{0.5em}}
\begin{document}
\baselineskip=16pt
\title{The conditional higher moment risk measure: second-order asymptotics with FGM contagion\thanks{This work  is supported by the National Social Science Fund of China (24BTJ034), the Provincial Natural Science Research Project of Anhui Colleges (2022AH050067, 2024AH050037) and the  doctoral research initiation fund of Anhui University (S020318033/015).}}

\author{Haifan Hu$^{a}$, Bingzhen Geng$^{a}$, Jiajun Liu$^{b}$ and Shijie Wang$^{a}$\\
{\small $^{a}$ School of Big Data and Statistics, Anhui University}\\
{\small $^{b}$ Department of Financial and Actuarial Mathematics, Xi'an Jiaotong-Liverpool University}
}

\date{}
\maketitle

\begin{abstract}
    This paper investigates second-order asymptotic expansions for the conditional higher moment ($\mathrm{CoHM}$) coherent risk measure under a Farlie–Gumbel–Morgenstern ($\mathrm{FGM}$) dependence structure, capturing a weak contagion between a primary loss risk and a reference risk. Assuming that the primary risk belongs to the Fr\'echet, Weibull, or Gumbel maximum domain of attraction, we systematically derive second-order asymptotic expansions using extreme value theory and second-order regular variation theory. Compared with existing first-order results, our refined approximations capture higher-order tail behavior and dependence effects more accurately. Numerical simulations confirm that the second-order asymptotics substantially reduce approximation errors, especially at extreme confidence levels. Empirical applications to insurance claim data further illustrate the practical superiority of the second-order approach.
\end{abstract}

\section{Introduction}
In the modern regulatory framework in insurance and finance, risk assessment fundamentally relies on coherent risk measures that capture extreme losses while satisfying desirable axiomatic properties. Value-at-Risk ($\mathrm{VaR}$) and Expected Shortfall ($\mathrm{ES}$) remain the cornerstones of Basel III and Solvency II, yet each has well-documented limitations: $\mathrm{VaR}$ fails to be subadditive and ignores the severity of losses beyond its quantile, while $\mathrm{ES}$, though coherent, provides only the first moment of tail losses and offers no flexibility to accommodate different degrees of risk aversion.

To overcome these shortcomings, 
\citet{Krokhmal2007} introduced the higher moment ($\mathrm{HM}$) coherent risk measure, which incorporates tail variability through an $L^k$ norm. By allowing the investor to choose a risk aversion parameter $k\ge 1$, the HM measure penalizes extreme events more heavily and nests $\mathrm{ES}$ as the special case $k=1$.
The $\mathrm{HM}$ risk measure can flexibly adjust sensitivity to extreme losses, allowing it to accommodate varying levels of risk aversion, and thus demonstrates broad applicability in portfolio optimization, capital allocation, and engineering reliability analysis. Subsequent theoretical investigations on $\mathrm{HM}$ risk measure have been conducted by \citet{Dentcheva2010} on Kusuoka representations of higher-order dual risk measures, by \citet{Chen2011} on post-optimality for mean-risk stochastic programs, and by \citet{Krokhmal2011} on modeling and optimization of risk. Whereas, for applications of $\mathrm{HM}$ risk measure, the reader is referred to \citet{Matmoura2013} for multistage option portfolio optimization, \citet{Vinel2017} for certainty equivalent measures of risk, \citet{Kouri2019} for higher moment buffered probability, and \citet{Gomez2022} for the gradient allocation principle based on the higher moment risk measure.

Despite the flexibility of unconditional $\mathrm{HM}$, real-world risk positions are rarely isolated. Consider two random variables (r.v.s) $X$ and $Y$, where $X$ represents a risk position and $Y$ a risk factor summarizing relevant sources of risk. Ignoring their dependence can severely underestimate tail risk. With a focus on extreme risks, we examine the conditional measurement of $X$ given the realization of an extreme scenario for $Y$. To this end, \cite{LiuYi2025} recently proposed a conditional version of the $\mathrm{HM}$ measure, denoted as $\mathrm{CoHM}_{p,q}(X|Y)$, where the conditioning event $\{Y > y_p\}$ captures such an extreme scenario. This measure combines three layers of information: the scenario risk level $p$ (how extreme the factor $Y$ is), the confidence level $q$ (tail probability of the primary loss $X$), and the risk aversion parameter $k$.

In a broad context, the above stylized framework for conditional risk measurement generalizes the unconditional $\mathrm{HM}$ risk measure by explicitly accounting for the additional risk contribution from a reference risk $Y$ to the primary loss risk $X$. From a systemic risk perspective, interpreting $Y$ as the extreme loss of the entire financial system and $X$ as the loss of an individual institution, $\mathrm{CoHM}$ measures the institution's conditional tail risk in the event of a systemic crisis. A large $\mathrm{CoHM}$ value signals a high vulnerability to systemic shocks, making it a useful tool to identify systemically important institutions and allocate systemic risk charges. From a general insurance perspective, $\mathrm{CoHM}$ acts as a conditional tail risk multiplier of one business line when another suffers a large loss, quantifying the tail contagion effect between lines. If $\mathrm{CoHM}$ is substantially larger than its unconditional counterpart, an extreme event amplifies the expected loss of the other line, providing a conditional capital multiplier that suggests increasing capital reserves under extreme scenarios. Furthermore, $\mathrm{CoHM}$ offers a pricing basis for cross-line triggered reinsurance (e.g., paying the other line contingent on a catastrophe in one line), thereby reflecting the reinsurer's true extreme exposure.

A large strand of post-crisis research has developed systemic risk measures that condition on extreme events. Prominent examples include $\mathrm{CoVaR}$ and $\mathrm{CoES}$ \citep{Adrian2016} as well as their further refinements \citep{Girardi2013, MainikSchaanning2014}, $\mathrm{SES}$ and $\mathrm{MES}$ \citep{Acharya2017}, and $\mathrm{SRISK}$ \citep{Acharya2012,Brownlees2017}. For comprehensive overviews, see \citet{Bisias2012}, \citet{Benoit2017} and  \citet{DeBandt2002}. We also refer to \citet{Asimit2018}, \citet{Qin2021}, \citet{Chen2022}, and \citet{LiuYi2025}, whose asymptotic frameworks are analogous to ours, as well as \citet{TXZ2026}, who investigated the $\mathrm{CoHM}$ risk measure under a nonstandard bivariate regular variation structure as $p\uparrow 1$ (with $q$ and $k$ fixed) and proposed asymptotically consistent estimators, thereby complementing the first‑order results of \citet{LiuYi2025} obtained under regime $q\uparrow 1$.

Most of the aforementioned conditional risk measures are based on a single reference risk, which does not reflect the reality that multiple risks may collapse simultaneously in financial markets. This limitation calls for multivariate systemic risk measures, which have attracted increasing attention; see, e.g., \citet{Lee2013}, \citet{Sun2018}, and \citet{Ling2019}. Following the idea of multi‑$\mathrm{CoVaR}$ introduced by \citet{Wen2025}, we further propose a multi-conditional higher moment ($\mathrm{MCoHM}$) risk measure. Instead of conditioning on an extreme event of a single reference risk, the $\mathrm{MCoHM}$ measure conditions on the joint extreme event of several reference risks. Concretely, for a given set of reference variables, each at its own scenario risk level, the $\mathrm{MCoHM}$ measure is defined as the minimal value of a threshold plus a scaled conditional $L^k$ norm of the excess loss, where the conditioning event requires all reference risks to exceed their respective high quantiles. This definition naturally generalizes the bivariate $\mathrm{CoHM}$ measure and helps capture risk contagion from multiple sources to the primary loss. Following the reasoning in \citet{LiuYi2025}, one can verify that the $\mathrm{MCoHM}$ measure is also coherent and shares analogous properties with its bivariate counterpart.

\citet{LiuYi2025} conducted a first-order asymptotic analysis of \(\mathrm{CoHM}_{p,q}(X|Y)\) as the confidence level \(q\) approaches one, assuming the primary risk belongs to the Fréchet, Weibull or Gumbel maximum domain of attraction ($\mathrm{MDA}$). Although these first-order results help regulators understand downside risk, they suffer from limited accuracy, especially at extreme confidence levels. With the growing demand for precise risk assessment, second-order (and higher-order) asymptotic analyses have become necessary. In this paper, we focus on the same FGM weak contagion framework and develop second-order asymptotic expansions for the \(\mathrm{CoHM}\) risk measure. We assume that the primary loss \(X\) lies in the MDA of the Fréchet, Weibull, or Gumbel distribution and that its tail satisfies a second-order regular variation condition. We emphasize that the FGM dependence structure is notable for its extensive applications in risk management and extreme value theory, owing to its analytical tractability and flexibility. It provides a practical tool for modeling the joint distributions of multiple loss variables, allowing portfolio risk assessment, extreme-event prediction, capital allocation, and enhanced risk defense capabilities. As demonstrated by \citet{WeiYuan2016}, the FGM distribution captures weak contagion without altering expected losses, while significantly affecting the tail distribution of portfolio losses, which is a desirable feature for low‑default portfolio risk management. These features have made the FGM copula a popular choice in quantitative risk management, as evidenced by studies on risk aggregation and capital allocation \citep{Cossette2013, BLIERWONG2023102}, ruin analysis under dependence \citep{Chen2015}, and second‑order asymptotics of risk concentration \citep{MaoYang2015}.

We derive explicit second-order expansions for the three $\mathrm{MDA}$ cases using extreme value theory and second-order regular variation theory. These expansions incorporate the FGM dependence parameter, marginal tail indices, and second-order parameters, characterizing how weak contagion and tail curvature influence the convergence rate. Relative to first-order results, our formulas substantially improve accuracy, often reducing approximation errors by substantial improvements in accuracy, as demonstrated by simulation studies. We further extend the analysis to the multi‑conditional ($\mathrm{MCoHM}$) setting, retaining the same second‑order structure. Empirical applications to light‑tailed motor insurance and heavy‑tailed fire loss data validate the FGM copula for modeling inter-line weak contagion, and the $\mathrm{CoHM}$ risk contribution metrics provide actionable benchmarks for differentiated regulatory capital allocation and reinsurance pricing.

The rest of this paper is organized as follows. Section \ref{CoHM} first introduces some definitions on $\mathrm{HM}$, $\mathrm{CoHM}$ and $\mathrm{MCoHM}$ risk measures, then gives the first-order asymptotics of $\mathrm{CoHM}$ risk measure. Section \ref{Second_asy} presents the second-order asymptotics of $\mathrm{CoHM}$ risk measure and subsequently conducts some in-depth discussions. Section \ref{Numerical_ill} gives numerical examples to show whether and how much the first-order expansions have been improved by the second-order analogs. Section \ref{Empirical_ana} provides an empirical analysis to illustrate the application of $\mathrm{CoHM}$ risk measure. Section \ref{Conclusion} concludes the work. Finally, Appendix \ref{Appendix} prepares some lemmas and provides all the proofs. 

\section{The CoHM risk measure}\label{CoHM}

\subsection{Notational conventions}

Let $(\Omega, \mathcal{F}, P)$ be a probability space. All r.v.s are real-valued and defined on this space. Throughout this paper, all limit relations are understood as $q \uparrow 1$ or $t \uparrow \infty$ unless otherwise specified.

For a r.v. $X$ with distribution function $F$, denote by $\overline{F}(x) = 1 - F(x)$ its survival function. For $q \in (0,1)$, the quantile function of $F$ is defined as

\[
F^{\leftarrow}(q) = \inf\{x \in \mathbb{R} : F(x) \geq q\},
\]
where, by convention, $\inf \emptyset$ is the right endpoint of the support set. The tail quantile function $U(\cdot)$ of $F$ is defined as

\begin{equation}\label{eq:quantile}
U(t) = \left(\frac{1}{\overline{F}}\right)^{\leftarrow}(t) = F^{\leftarrow}\left(1 - \frac{1}{t}\right), \qquad t \geq 1. 
\end{equation}

For two positive functions $f$ and $g$, we write $f \sim g$ if $\lim f/g = 1$, and $f = o(g)$ if $\lim f/g = 0$. For any $x, y \in \mathbb{R}$, set $x \vee y = \max\{x, y\}$.

\subsection{Definition of the HM risk measure}\label{sec:2.2}

Let $X$ be a real-valued r.v. representing a risk position in the loss-profit style, and let $F$ denote its distribution function. To avoid trivialities, we assume throughout that $X$ is not degenerate at a constant. Following \citet{Krokhmal2007}, for $X \in L^k$ with $k \geq 1$, the higher moment ($\mathrm{HM}$) coherent risk measure of $X$ at confidence level $0 < q < 1$ is defined as
\begin{eqnarray}
\mathrm{HM}_q(X) = \inf_{x \in \mathbb{R}} \left\{ x + \frac{1}{1 - q} \left\| (X - x)_+ \right\|_{L^k} \right\}, \label{eq:HM}
\end{eqnarray}
where $\| (X - x)_+ \|_{L^k} = \left(E\left[(X - x)_+^k\right]\right)^{1/k}$ denotes the $L^k$ norm, and $(X - x)_+ = \max\{X - x, 0\}$.

The effective region for the infimum in \eqref{eq:HM} is $-\infty < x \leq \hat{x}$, where $\hat{x} = \inf\{y \in \mathbb{R} : P(X \leq y) = 1\}$ denotes the essential supremum of $X$. Moreover, if $\hat{x} = \infty$, or if $\hat{x} < \infty$ and $\hat{\pi} = P(X = \hat{x}) < (1 - q)^k$, then the effective region reduces to $-\infty < x < \hat{x}$. Following \citet{Bellini2012}, any minimizer $x^* = x_q^* < \hat{x}$ of \eqref{eq:HM} is called an Orlicz quantile of $X$ at the level $q$.

When $k = 1$, the $\mathrm{HM}$ risk measure coincides with the well-known $\mathrm{ES}$ at the confidence level $q$:
\[
\mathrm{HM}_q(X) = \mathrm{ES}_q(X) = x^* + \frac{1}{1 - q} E\left[(X - x^*)_+\right],
\]
where any minimizer $x^*$ belongs to the interval $[F^{\leftarrow}(q), F^{\rightarrow}(q)]$ with $F^\rightarrow$ denoting the right generalized inverse of $F$ defined as $F^\rightarrow(q)=\sup\{x\in\mathbb{R}:F(x)\leq q\}.$ This interval may collapse to a single point unless $q$ corresponds to a flat portion of the distribution $F$.

For $k > 1$, assuming either $\hat{x} = \infty$ or $q \in (0, 1 - \hat{\pi}^{1/k})$ when $\hat{x} < \infty$, the HM risk measure admits the explicit representation
\[
\mathrm{HM}_q(X) = x^* + \frac{1}{1 - q} \left\| (X - x^*)_+ \right\|_{L^k},
\]
where $x^* = x_q^* < \hat{x}$ is the unique minimizer of \eqref{ratio_solution}, uniquely determined by the first-order condition

\begin{eqnarray}
\frac{\left(E\left[(X - x)_+^{k-1}\right]\right)^k}{\left(E\left[(X - x)_+^k\right]\right)^{k-1}} = (1 - q)^k. \label{ratio_solution}
\end{eqnarray}

The strict inequality $x^* < \hat{x}$ follows from Lemma \ref{lem:first+0} in the Appendix \ref{Appendix}. For further details, we refer to Example 2.2 of \citet{Krokhmal2007}, equation (22) of \citet{Dentcheva2010}, and Theorem 2.1 of \citet{TangYang2012}.

Beyond its coherence and law invariance, the HM risk measure offers considerable flexibility in risk assessment through its two key parameters: the confidence level $q \in (0,1)$, which governs the tail region under consideration, and the risk aversion parameter $k \geq 1$, which controls the sensitivity to extreme outcomes and ensures compatibility with higher-order stochastic dominance. The interplay between $q$ and $k$ allows the HM measure to blend the features of monetary risk measures (through $q$) and variability risk measures (through $k$), making it a versatile tool for capturing both loss magnitude and tail uncertainty. We refer to Subsection 2.4 of \citet{Krokhmal2007}, Subsections 2.2 and 2.4 of \citet{Gomez2022}, and Section 4 of \citet{Pichler2024} for detailed justifications of these properties.

\begin{remark}\label{rem:HM-HG}
The Haezendonck-Goovaerts ($\mathrm{HG}$) risk measure, introduced by \citet{Haezendonck1982} within the framework of Orlicz-norm premium principles, is closely related to the HM risk measure. Specifically, when the Young function is chosen as $\phi(t) = t^k$ for $k \geq 1$, the $\mathrm{HG}$ risk measure at level $q$ coincides with the HM risk measure at an adjusted confidence level. More precisely, for any integrable r.v. $Z$, one has

\[
\mathrm{HG}_q(Z) = \mathrm{HM}_{\tilde{q}}(Z), \qquad \tilde{q} = 1 - (1 - q)^{1/k},
\]
or equivalently,
\[
\mathrm{HM}_q(Z) = \mathrm{HG}_{1 - (1 - q)^k}(Z).
\]
This equivalence will be exploited in Section 3 to derive the second-order asymptotic expansions of the $\mathrm{CoHM}$ risk measure by leveraging known results for the $\mathrm{HG}$ risk measure.
\end{remark}

\subsection{Definition of the CoHM risk measure} \label{def_CoHM}

Consider a pair of real-valued random vectors $(X,Y)$ with marginal distribution functions $F$ and $G$, respectively. Assume $X \in L^k$ for some $k \geq 1$. Let $y_p = G^{\leftarrow}(p)$ denote the $p$-th quantile of $Y$ for $0 < p < 1$, and suppose $P(Y > y_p) > 0$. The conditional higher moment ($\mathrm{CoHM}$) risk measure of $X$ at confidence level $0 < q < 1$, with risk aversion parameter $k \geq 1$, and conditional on the extreme event $\{Y > y_p\}$, is defined as

\begin{eqnarray}
\mathrm{CoHM}_{p,q}(X|Y) = \inf_{x \in \mathbb{R}} \left\{ x + \frac{1}{1 - q} \left( E\left[(X - x)_+^k \mid Y > y_p\right] \right)^{1/k} \right\}. \label{eq:CoHM}
\end{eqnarray}

In \eqref{eq:CoHM}, the conditioning event $\{Y > y_p\}$ with $p$ close to 1 captures an extreme scenario of the reference risk $Y$. The CoHM measure thus integrates three layers of information: the scenario risk level $p$, which determines the extremity of the conditioning event; the confidence level $q$, which specifies the tail probability of the primary loss $X$; and the risk aversion parameter $k$, which governs the sensitivity to extreme outcomes.

The effective region for the infimum in \eqref{eq:CoHM} is identical to that of the unconditional HM risk measure discussed in Subsection 2.2, namely $-\infty < x \leq \hat{x}$, where $\hat{x}$ denotes the essential supremum of $X$. If $\hat{x} = \infty$, or if $\hat{x} < \infty$ and $\hat{\pi} = P(X = \hat{x}) < (1 - q)^k$, then the effective region becomes $-\infty < x < \hat{x}$. Following the terminology of \citet{Bellini2012}, any minimizer $x^* = x_{p,q}^* < \hat{x}$ of \eqref{eq:CoHM} is called a conditional Orlicz quantile of $X$ at the level $q$ given $\{Y > y_p\}$.

When $k = 1$, the CoHM risk measure reduces to the conditional expected shortfall ($\mathrm{CoES}$):
\[
\mathrm{CoHM}_{p,q}(X|Y) = \mathrm{CoES}_{p,q}(X|Y) = x^* + \frac{1}{1 - q} E\left[(X - x^*)_+ \mid Y > y_p\right],
\]
where any minimizer $x^*$ lies in the interval $[H_p^{\leftarrow}(q), H_p^{\rightarrow}(q)]$, with $H_p$ denoting the conditional distribution of $X$ given $\{Y > y_p\}$. 

For $k > 1$, the minimizer $x^*$ is unique and characterized by the first-order condition
\[
\frac{\left(E\left[(X-x)_+^{k-1} \mid Y > y_p\right]\right)^k}{\left(E\left[(X-x)_+^k \mid Y > y_p\right]\right)^{k-1}} = (1 - q)^k. 
\]

A key observation, which will be instrumental in our asymptotic analysis, is that the $\mathrm{CoHM}$ risk measure can be expressed as an unconditional HM risk measure under a modified probability measure. Indeed, for each fixed $p \in (0,1)$, define the conditional probability measure
\[
\tilde{P}_p(A) = P(A \mid Y > y_p), \qquad A \in \mathcal{F}. 
\]
Then \eqref{eq:CoHM} can be rewritten as
\[
\mathrm{CoHM}_{p,q}(X|Y) = \inf_{x \in \mathbb{R}} \left\{ x + \frac{1}{1 - q} \left( E^{\tilde{P}_p}\left[(X - x)_+^k\right] \right)^{1/k} \right\} = \mathrm{HM}_q^{\tilde{P}_p}(X), 
\]
where the superscript $\tilde{P}_p$ indicates that expectations are taken under the conditional measure $\tilde{P}_p$. Equivalently, treating the conditional distribution $X \mid \{Y > y_p\}$ as a probability distribution in its own right, we have the fundamental identification
\begin{eqnarray}
\mathrm{CoHM}_{p,q}(X|Y) = \mathrm{HM}_q(X \mid Y > y_p). \label{eq:CoHM-as-HM}
\end{eqnarray}
The representation $\mathrm{CoHM}_{p,q}(X|Y) = \mathrm{HM}_q^{\tilde{P}_p}(X)$ in \eqref{eq:CoHM-as-HM} immediately implies that the $\mathrm{CoHM}$ risk measure inherits the coherence, law invariance, and Orlicz-quantile characterization of the unconditional $\mathrm{HM}$ risk measure. This observation is central to our analysis, as it allows the direct application of existing second-order results for the HM risk measure to the conditional setting, with the conditional measure $\tilde{P}_p$ playing the role of the reference probability.

\begin{remark}\label{rem:CoHM-asymptotic-regimes}
The $\mathrm{CoHM}$ risk measure defined in \eqref{eq:CoHM} is amenable to two distinct asymptotic regimes, each corresponding to a different source of extremity. The first regime, studied by \citet{LiuYi2025} and further refined in the present paper, considers $q \uparrow 1$ with $p$ fixed, capturing the extremity of the primary loss $X$ given a fixed extreme scenario of $Y$. The second regime, investigated by \citet{TXZ2026}, considers $p \uparrow 1$ with $q$ fixed, capturing the effect of an increasingly extreme conditioning event on the risk assessment of $X$. Our second-order expansions in Section 3.1 focus on the first regime, while Section \ref{Joint_asy} discusses the interplay between these two asymptotic perspectives and the iterated limit structure.
\end{remark}

Following the concept of risk contribution measures introduced by \citet{Adrian2016} in the context of $\mathrm{CoVaR}$ and further developed by \citet{Zhang2024} for distortion risk measures, we introduce two complementary risk contribution measures associated with the $\mathrm{CoHM}$ coherent risk measure. The first is the absolute risk contribution
\begin{equation}\label{eq:Delta}
\Delta \mathrm{CoHM}_{p,q}(X,Y) = \mathrm{CoHM}_{p,q}(X|Y) - \mathrm{HM}_q(X),
\end{equation}
which measures the incremental tail risk attributable to conditioning on the extreme event $\{Y > y_p\}$. To facilitate comparison across different scales of risk, we also introduce the relative risk contribution
\begin{equation}\label{eq:DeltaR}
\Delta \mathrm{CoHM}_{p,q}^{R}(X,Y) =
\begin{cases}
\dfrac{\Delta \mathrm{CoHM}_{p,q}(X,Y)}{\mathrm{HM}_q(X)}, & \text{if } \hat{x} = \infty, \\[0.8cm]
\dfrac{\Delta \mathrm{CoHM}_{p,q}(X,Y)}{\hat{x} - \mathrm{HM}_q(X)}, & \text{if } \hat{x} < \infty,
\end{cases} 
\end{equation}
provided that $\mathrm{HM}_q(X), \hat{x}-\mathrm{HM}_q(X) > 0$, respectively. This relative measure captures the proportional amplification of the risk position due to the extreme scenario of $Y$, and is particularly useful when the unconditional risk measure $\mathrm{HM}_q(X)$ varies substantially across different assets or market conditions.

The absolute contribution $\Delta \mathrm{CoHM}_{p,q}(X,Y)$ in \eqref{eq:Delta} serves as a direct analog of the $\Delta \mathrm{CoVaR}$ measure of \citet{Adrian2016}, while the relative contribution $\Delta \mathrm{CoHM}_{p,q}^{R}(X,Y)$ in \eqref{eq:DeltaR} provides a dimensionless counterpart similar to the systemic risk ratio studied by \citet{Zhang2024}. Both quantities will be analyzed asymptotically in subsection \ref{subsection First-order asymptotics}, where we derive explicit first-order expansions that reveal their sensitivity to the FGM dependence parameter $r$ and the tail index $\alpha$.

\subsection{Extension to the multivariate case: MCoHM}\label{multivariate CoHM}

Extending the bivariate case, consider a primary loss $X$ and $n$ reference risks $Y_1, \ldots, Y_n$ with marginal distribution functions $G_1, \ldots, G_n$, respectively. For a confidence level $q \in (0,1)$ and a vector $\mathbf{p} = (p_1, \ldots, p_n) \in (0,1)^n$, define $\mathbf{y}_{\mathbf{p}} = (G_1^{\leftarrow}(p_1), \ldots, G_n^{\leftarrow}(p_n))$. The multi-conditional higher moment ($\mathrm{MCoHM}$) risk measure is defined as
\begin{eqnarray}
\mathrm{MCoHM}_{\mathbf{p},q}(X|\mathbf{Y}) = \inf_{x \in \mathbb{R}} \left\{ x + \frac{1}{1 - q} \left( E\left[(X - x)_+^k \mid \mathbf{Y} > \mathbf{y}_{\mathbf{p}}\right] \right)^{1/k} \right\}, 
\end{eqnarray}
where inequality $\mathbf{Y} > \mathbf{y}_{\mathbf{p}}$ is understood componentwise. This measure conditions on the joint extreme event of several reference risks and naturally generalizes the bivariate $\mathrm{CoHM}$ measure. Following the reasoning in \citet{LiuYi2025} and the discussion in Subsection \ref{def_CoHM}, one can verify that $\mathrm{MCoHM}$ is coherent and shares analogous properties with its bivariate counterpart. The second-order asymptotic analysis of this multivariate extension is developed in Section \ref{MCOHM}.

\subsection{First-order asymptotics}\label{subsection First-order asymptotics}
In extreme value theory, the concept of max-domains of attraction plays a fundamental role in characterizing the tail behavior of a distribution. A distribution function \(F\) with an upper endpoint \(\hat{x}\le\infty\) is said to belong to the $\mathrm{MDA}$ of an extreme value distribution \(W\), denoted by \(F\in\mathrm{MDA}(W)\), if there exist some normalizing constants \(c_n>0\) and \(d_n\in\mathbb{R}\) such that \(F^n(c_nx+d_n)\) converges weakly to \(W(x)\) as \(n\to\infty\). The famous Fisher-Tippett theorem (see \citet{Fisher1928}) states that $W$ has to be one of the Fr\'echet, Weibull, and Gumbel distributions whose standard forms are given by, respectively,
\begin{align}
\mbox{Fr\'echet case:} &\quad\quad \quad \quad  \displaystyle \Phi_\alpha(x)=\exp\{-x^{-\alpha}\}, \quad\quad \quad \quad  x>0,~\alpha>0,\nonumber\\
\mbox{Weibull case:} &\quad\quad \quad \quad \displaystyle \Psi_\alpha(x)=\exp\{-|x|^\alpha\}, \quad\quad \quad \quad  x\leq 0,~\alpha>0,\nonumber\\
\mbox{Gumbel case:} &\quad\quad \quad \quad \displaystyle \Lambda(x)=\exp\{-e^{-x}\}, \quad\quad \quad \quad ~~~ x\in\mathbb{R}.\nonumber
\end{align}

For the Fr\'echet case, due to Theorem 3.3.7 of \citet{Embrechts1997}, a distribution \(F\) belongs to \(\mathrm{MDA}(\Phi_\alpha)\) if and only if its upper endpoint \(\hat{x}=\infty\) and its tail is regularly varying with the index \(-\alpha\), written as $\overline{F}\in\mathrm{RV}_{-\alpha}$, that is,
\begin{equation}\label{eq:4}
\lim_{x\to\infty}\frac{\overline{F}(xy)}{\overline{F}(x)}=y^{-\alpha},\qquad y>0.\nonumber
\end{equation}
 One can easily see that \(\overline{F}\in\mathrm{RV}_{-\alpha}\) implies that the tail of $F$ decays like a power law, which is typical for heavy‑tailed risks such as Pareto distributions.

For the Weibull case, it follows from Theorem 3.3.12 of \citet{Embrechts1997} that a distribution \(F\) belongs to \(\mathrm{MDA}(\Psi_\alpha)\) if and only if its upper endpoint \(\hat{x}\) is finite, and
\begin{equation}\label{eq:5}
\lim_{x\to \infty}\frac{\overline{F}\left(\hat{x}-\frac{1}{xy}\right)}{\overline{F}\left(\hat{x}-\frac{1}{x}\right)}=y^{-\alpha},\qquad y>0.\nonumber
\end{equation}
Clearly, the above relation is equivalent to \(\overline{F}\Big(\hat{x}-\frac{1}{\cdot}\Big)\in\mathrm{RV}_{-\alpha}\). Thus, \(\mathrm{MDA}(\Psi_\alpha)\) always describes distributions with a finite right endpoint, such as the Beta distribution.

For the Gumbel case, by Theorem 3.3.27 of \citet{Embrechts1997}, a distribution \(F\) with \(\hat{x}\le\infty\) belongs to \(\mathrm{MDA}(\Lambda)\) if and only if there exists a positive auxiliary function \(a(\cdot)\) on $(-\infty, \hat {x})$ such that
\begin{equation}\label{eq:6}
\lim_{x\uparrow\hat{x}}\frac{\overline{F}(x+ya(x))}{\overline{F}(x)}=e^{-y},\qquad y\in\mathbb{R}.   
\end{equation}
Commonly, the auxiliary function can be chosen as the mean excess function \(a(x)=\mathbb{E}[X-x\mid X>x]\) that satisfies \(a(x)=o(x)\) if \(\hat{x}=\infty\) and \(a(x)=o(\hat{x}-x)\) if \(\hat{x}<\infty\). 

For the $\mathrm{CoHM}$ risk measure defined in (\ref{eq:CoHM})  with a bivariate risk vector \((X,Y)\)  that follows a FGM joint distribution in \eqref{eq:9}.  The bivariate FGM distribution was
initially proposed by \citet{M1956} and subsequently investigated by \citet{F1960} and \citet{G1960}. Specifically, let \(F\) and \(G\) be the marginal distribution functions of  \(X\) and \(Y\), respectively. The joint distribution function of the FGM family is of the form
\begin{equation}\label{eq:9}   
\Pi(x, y) = F(x) G(y) \left[ 1 + r \left(1 - F(x)\right) \left(1 - G(y)\right) \right], \quad r \in [-1, 1],  
\end{equation}
where the parameter $r$ measures the dependence between \(X\) and \(Y\). This structure captures weak contagion since it does not alter the expected losses but affects the tail of the portfolio loss.

 Under this framework, we can derive the first-order asymptotics results from Theorems 3.1-3.3 of \citet{LiuYi2025}. Hereafter, denote by $\Gamma(\cdot)$ the Gamma function, that is, $\Gamma(x)=\int_0^\infty t^{x-1}e^{-t}dt$ for $x>0$. Denote by $B(\cdot,\cdot)$ the Beta function, that is, $B(a,b)=\int_0^1x^{a-1}(1-x)^{b-1}dx$ for $a,b>0$.  The following first-order relations hold as $q\uparrow1$:

\subsubsection*{Fr\'echet Case}
If $F \in \mathrm{MDA}(\Phi_\alpha)$ with $\alpha > k$, then

\[
\mathrm{CoHM}_{p,q}(X|Y) \sim c_1 m_1^{1/\alpha} F^{\leftarrow}(1 - (1 - q)^k), 
\]
where
\begin{equation}\label{eq:c_1} 
c_1 = \frac{\alpha(\alpha - k)^{k/\alpha - 1}}{k^{(k-1)/\alpha}} \left(B(\alpha - k, k)\right)^{1/\alpha}, \qquad 
\end{equation}
 and $m_1 = 1 + rG(y_p).$ The corresponding unconditional $\mathrm{HM}$ risk measure is obtained by setting $r = 0$ (i.e., $m_1 = 1$), yielding
\[
\mathrm{HM}_q(X) \sim c_1 F^{\leftarrow}(1 - (1 - q)^k). 
\]
Consequently, for $r\neq 0,$ the first-order asymptotic of the risk contribution $\Delta \mathrm{CoHM}_{p,q}(X,Y)$ defined in (\ref{eq:Delta}) is given by
\begin{eqnarray*}
\Delta \mathrm{CoHM}_{p,q}(X,Y) \sim c_1 \left(m_1^{1/\alpha} - 1\right) F^{\leftarrow}(1 - (1 - q)^k). 
\end{eqnarray*}
This reveals that the amplification effect due to conditioning on the extreme event of $Y$ is proportional to the quantile $F^{\leftarrow}(1 - (1 - q)^k)$ and scales with the factor $(m_1^{1/\alpha} - 1)$, which depends on the FGM dependence parameter $r$ and the reference risk level $p$ through $G(y_p)$.
Moreover, using \eqref{eq:DeltaR} and the first-order asymptotics above, the relative risk contribution admits the simple limiting form

\begin{equation}
\Delta \mathrm{CoHM}_{p,q}^{R}(X,Y) \sim m_1^{1/\alpha} - 1. \label{DeltaCoHM_F}
\end{equation}

This result is particularly striking: in the Fr\'echet case, the relative risk contribution converges asymptotically to a constant that depends only on the FGM parameter $r$ and the reference risk level $p$, independent of the confidence level $q$. This asymptotic constancy implies that for sufficiently extreme levels of primary loss $X$, the proportional amplification of risk due to the conditioning of the extreme event of $Y$ stabilizes at a value determined solely by the tail index $\alpha$ and the dependence structure. This provides a robust, scale-invariant benchmark for systemic risk assessment that does not require precise knowledge of the tail quantile $F^{\leftarrow}(1 - (1 - q)^k)$.
\subsubsection*{Weibull Case}
If $F \in \mathrm{MDA}(\Psi_\alpha)$ with $\alpha > 0$ and finite upper endpoint $\hat{x}$, then
\[
\hat{x} - \mathrm{CoHM}_{p,q}(X|Y) \sim c_2 m_1^{-1/\alpha} \left(\hat{x} - F^{\leftarrow}(1 - (1 - q)^k)\right),
\]
where
\begin{equation}\label{eq:c_2} 
c_2 = \frac{\alpha k^{(k-1)/\alpha}}{(\alpha + k)^{k/\alpha + 1}} \left(B(\alpha + 1, k)\right)^{-1/\alpha}.
\end{equation}
The unconditional $\mathrm{HM}$ risk measure satisfies
\[
\hat{x} - \mathrm{HM}_q(X) \sim c_2 \left(\hat{x} - F^{\leftarrow}(1 - (1 - q)^k)\right). 
\]
Therefore, for $r\neq 0,$ the first-order asymptotic of the risk contribution $\Delta \mathrm{CoHM}$ is
\begin{eqnarray*}
\Delta \mathrm{CoHM}_{p,q}(X,Y) \sim c_2 \left(1-m_1^{-1/\alpha} \right) \left(\hat{x}-F^{\leftarrow}(1 - (1 - q)^k)\right). 
\end{eqnarray*}
Since $F^{\leftarrow}(1 - (1 - q)^k) \uparrow \hat{x}$ as $q \uparrow 1$, the term $\left(F^{\leftarrow}(1 - (1 - q)^k) - \hat{x}\right)$ is positive and so is $\Delta \mathrm{CoHM}_{p,q}(X,Y)$ because $1-m_1^{-1/\alpha} >0$ when $r > 0$. This confirms that conditioning on the extreme event of $Y$ reduces the distance to the upper endpoint, thus increasing the risk measure. For the relative risk contribution, we obtain
\begin{equation}
\Delta \mathrm{CoHM}_{p,q}^{R}(X,Y) \sim 1- m_1^{-1/\alpha} . \label{DeltaR_W}
\end{equation}
Here, the amplification factor is determined by the same dependence structure but enters with a negative exponent, reflecting the bounded nature of Weibull-type distributions. When $r > 0$, the limiting value is positive, indicating that the risk contribution relative to the unconditional measure increases (i.e., the distance to the endpoint shrinks proportionally), consistent with the intuition that positive dependence amplifies tail risk. When $r < 0$, the relative contribution decreases, reflecting the risk-reducing effect of negative dependence.

\subsubsection*{Gumbel Case}
If $F \in \mathrm{MDA}(\Lambda)$ with upper endpoint $\hat{x} \leq \infty$, then for $\hat{x} = \infty$,
\[
\mathrm{CoHM}_{p,q}(X|Y) \sim F^{\leftarrow}\left(1 -(1 - q)^k \right), 
\]
and for $\hat{x} < \infty$,
\[
\hat{x} - \mathrm{CoHM}_{p,q}(X|Y) \sim \hat{x} - F^{\leftarrow}\left(1 - (1 - q)^k\right).
\]
The unconditional $\mathrm{HM}$ risk measure is obtained for $\hat{x} = \infty$,
\[
\mathrm{HM}_q(X) \sim F^{\leftarrow}\left(1 -(1 - q)^k \right). 
\]
\begin{remark}
Indeed, if $\hat{x} = \infty$, following Theorem 3.2 (i) of \citet{LiuYi2025}, then
\[
\mathrm{CoHM}_{p,q}(X|Y) \sim F^{\leftarrow}\left(1 - m_1^{-1} \frac{k^k}{\Gamma(k+1)}(1 - q)^k\right). 
\]
Observe that
\[
F^{\leftarrow}\left(1 - m_1^{-1} \frac{k^k}{\Gamma(k+1)}(1 - q)^k\right)=U\left(m_1 \frac{\Gamma(k+1)}{k^k}(1 - q)^{-k}\right).
\]
Since $F\in\mathrm{MDA}(\Lambda)$ implies that $U\in\mathrm{ERV}_0$ with an auxiliary function $a(\cdot)$, we have
\begin{align*}
F^{\leftarrow}\left(1 - m_1^{-1} \frac{k^k}{\Gamma(k+1)}(1 - q)^k\right)&=U\left((1-q)^{-k}\right)+a\left((1-q)^{-k}\right)\log\left(m_1\frac{\Gamma(k+1)}{k^k}\right)(1+o(1))\\
&=F^{\leftarrow}\left(1 -(1 - q)^k \right)(1+o(1)),
\end{align*}
where the last step holds due to $a(\cdot)=o(U(\cdot))$ when $\hat{x} = \infty$. Therefore, from this point of view, our result is in accordance with that of \citet{LiuYi2025}. For the case of $\hat{x}<\infty$, it follows from Theorem 3.2 (ii) of \citet{LiuYi2025} that
\[
\hat{x} - \mathrm{CoHM}_{p,q}(X|Y) \sim \hat{x} - F^{\leftarrow}\left(1 - m_1^{-1} \frac{k^k}{\Gamma(k+1)}(1 - q)^k\right).
\]
By an argument similar to above, one can conclude that our result is also in accordance with that of \citet{LiuYi2025}.  This leads to an interesting fact: first-order asymptotics of $\mathrm{CoHM}_{p,q}(X\mid Y)$ are insensitive to FGM dependence structure.

\end{remark}

Unlike the Fr\'echet and Weibull cases, the Gumbel case does not yield a simple constant limit for risk contribution measures. Lemma A.5 of \citet{LiuYi2025} shows that $F^{\leftarrow}(1-\cdot)$ is slowly varying at $0_+$. Consequently, we derive that
\begin{equation}
\lim_{q\uparrow1}\Delta \mathrm{CoHM}_{p,q}^{R}(X,Y) =\lim_{q\uparrow1}\frac{\mathrm{CoHM}_{p,q}(X|Y) - \mathrm{HM}_q(X)}{\mathrm{HM}_q(X)}=0. \label{DeltaR_G}
\end{equation}
Thus, in the Gumbel case, the relative risk contribution vanishes asymptotically, indicating that the conditioning effect is of a smaller order than the unconditional risk measure itself. For $\hat{x} < \infty$, a similar argument with the appropriate auxiliary function yields the same qualitative conclusion.

\begin{remark}
The first-order asymptotic formulas for the absolute and relative risk contributions in \eqref{DeltaCoHM_F}--\eqref{DeltaR_G} reveal a fundamental distinction among the three MDA classes. In the heavy-tailed Fr\'echet case and the bounded Weibull case, the relative risk contribution $\Delta \mathrm{CoHM}_{p,q}^{R}(X,Y)$ converges to a non-trivial constant that depends only on the FGM dependence parameter $r$, the reference risk level $p$, and the tail index $\alpha$. This constancy provides a scale-invariant benchmark that is robust to the exact choice of the confidence level $q$ for the extreme $q$. In contrast, for the light-tailed Gumbel case, the relative contribution vanishes asymptotically, reflecting the slower growth of the quantile function relative to the auxiliary function. These results have important practical implications: in heavy-tailed settings, the proportional amplification due to systemic risk conditioning can be reliably estimated from the tail index and dependence parameters alone, without requiring precise estimation of high quantiles.
These first-order asymptotics of $\Delta \mathrm{CoHM}$ and $\Delta \mathrm{CoHM}^{R}$ serve as a foundation for the refined second-order expansions developed in Section \ref{Second_asy}, which further incorporate tail curvature effects through the second-order regular variation parameter $\rho$ and provide more accurate approximations at finite confidence levels.
\end{remark}

\section{Second-order asymptotics}\label{Second_asy}

In order to further investigate the second-order asymptotic expansions for the $\mathrm{CoHM}$ coherent risk measure, we recall some concepts of first-order and second-order regular variations and extended regular variation for later use. The concept of first-order regular variation is used to characterize the three MDAs. This notion of first-order extended regular variation was first introduced by \cite{deHaan1996} to unify various limit results in EVT. The concepts of second-order regular variation are introduced to study the convergence speed of first-order regular variation. For some detailed properties, the reader is referred to \cite{deHaan2006} and references therein.

\begin{definition}\label{def:1}
A positive measurable function $f (\cdot)$ defined on \((0,\infty)\) is said to be first-order regularly varying at infinity with index \(\alpha \in \mathbb{R}\), denoted by \(f(\cdot) \in \mathrm{RV}_\alpha\), if for all \(x > 0\),
\[\lim_{t \to \infty} \frac{f(tx)}{f(t)} = x^\alpha.\]
Specifically, if \(\alpha = 0\), then \(f\) is said to be slowly varying. 
\end{definition}\label{def:2}
\begin{definition}
 A positive measurable function $f (\cdot)$ is said to belong to the class of first-order extended regularly varying functions with index \(\alpha \in \mathbb{R}\), denoted by \(f(\cdot) \in \mathrm{ERV}_\alpha\), if there exists a positive auxiliary function \(a(\cdot) > 0\) such that, for all \(x>0\),
\begin{equation}\label{eq:7}
\lim_{t \to \infty} \frac{f(tx) - f(t)}{a(t)} = \frac{x^\alpha-1}{\alpha}, 
\end{equation}
where the right-hand side above is interpreted as $\log x$ when $\alpha = 0$.
\end{definition}
\begin{remark}
 By Theorem 1.1.6 of \citet{deHaan2006}, it is proven that $F\in\mathrm{MDA}(W_\alpha)$ if and only if $U\in\mathrm{ERV}_{\alpha}$, where
 \[
W_{\alpha} = \begin{cases}
\displaystyle  \Phi_{1/\alpha}, & \alpha>0,\\[0.4em]
\displaystyle  \Psi_{-1/\alpha}, & \alpha<0,\\[0.4em]
\displaystyle \Lambda, & \alpha=0.
\end{cases}
\]
\end{remark}
\begin{definition}\label{def:3}
A positive measurable function $f (\cdot)$ is said to be second-order regularly varying with the first-order index $\alpha\in\mathbb{R}$ and second-order index $\rho\leq 0$, denoted by $f (\cdot)\in 2\mathrm{RV}_{\alpha,\rho}$, if there exists an auxiliary function $A(\cdot)$, which does not change sign eventually and converges to zero, such that, for all \(x > 0\),
\begin{equation}\label{eq:8}
\lim_{t \to \infty} \frac{ \frac{f(tx)}{f(t)} - x^\alpha }{A(t)} = x^\alpha \frac{x^\rho - 1}{\rho}=:H_{\alpha,\rho}(x),
\end{equation}
where, by convention, \((x^\rho - 1)/\rho \) is also understood as \(\log x\) when \(\rho = 0\) similarly as above. 
\end{definition}

\begin{remark}
If the functions $a(\cdot)$ and $A(\cdot)$ satisfy (\ref{eq:7}) and (\ref{eq:8}), respectively, then it is known that $a\in\mathrm{RV}_\alpha$ and $|A|\in\mathrm{RV}_\rho$; see also \cite{MaoHu2012} for details.  
\end{remark}
\begin{remark}
 It is worth mentioning that the convergences in Definitions \ref{def:1}-\ref{def:3} are all uniform
with respect to $x$ in any compact subset of $\mathbb{R}^+$; see, for example, Theorem B.1.4 and Theorem
B.2.9 of \citet{deHaan2006}.
\end{remark}




\subsection{A bivariate case}\label{subsec:3.1}
This subsection presents second-order asymptotic expansions for the $\mathrm{CoHM}$ coherent risk measure where the primary risk $X$ and the reference risk $Y$ are bivariate FGM dependent and the primary risk $X$ belongs to the Fr\'echet, Weibull or Gumbel MDA, respectively. Moreover,  for later notational convenience, in what follows, we use the following notation:
\begin{align*}
 m_{q}&=1+rG(y_p)\left(1-(1-q)^k\right), \\
\mathcal{J}_{\alpha, \rho, k}^{(1)} &= \frac{\alpha}{\rho} \left[ k^{\frac{\rho}{\alpha}}\left( \frac{\alpha - k}{k} \right)^{\frac{k\rho}{\alpha}}  \left(B\left(\alpha - k, k\right)\right)^{\frac{\rho}{\alpha} - 1} 
(B\left(\alpha - k - \rho, k\right)-1 \right], \\
\mathcal{J}_{\alpha,\rho,k}^{(2)} &= \frac{\alpha}{\rho} \left[ \alpha^{\frac{\rho}{\alpha}}\left( \frac{\alpha + k}{k} \right)^\frac{(k-1){\rho}}{\alpha}  \frac{(\alpha - \rho)(\alpha + k)}{\alpha(\alpha - \rho + k)} 
\left(B\left(\alpha, k\right)\right)^{\frac{\rho}{\alpha} - 1} 
B\left(\alpha - \rho, k\right)-1 \right],\\
\mathcal{M}^{(i)}&= \frac{m_1^\frac{\rho}{\alpha}-1}{\rho\alpha}+ \frac{m_1^{\frac{\beta}{\alpha}}}{\alpha^2}\mathcal{J}_{\alpha, \rho, k}^{(i)},~~~
\mathcal{N}^{(i)}= \left(\frac{m_1^\frac{\rho}{\alpha}-1}{\rho m_1}+ \frac{1}{\alpha m_1^2}\mathcal{J}_{\alpha, \rho, k}^{(i)}\right)rG(y_p),~~~i=1,2.
\end{align*}

\begin{theorem}\label{thm:main}
Consider the $\mathrm{CoHM}$ risk measure defined in (\ref{eq:CoHM}) for some $k\geq 1$ with a bivariate risk vector \((X,Y)\)  that satisfies the FGM dependence given by (\ref{eq:9}). Assume that \( F \in \mathrm{MDA}(\Phi_\alpha) \) for some $\alpha>k$ with \( P(X = \hat{x}) = 0 \). If, further, \( \overline{F} \in 2\mathrm{RV}_{-\alpha, \beta} \) for some $\beta\leq 0$ with an auxiliary function \( A(\cdot) \), then, as $q\uparrow 1$,
\begin{flalign*}
\mathrm{CoHM}_{p,q}(X|Y) &= c_1 m_{q}^{\frac{1}{\alpha}}
F^{\leftarrow}\left(1 - (1 - q)^k\right)\\
&\quad\times\left[ 1 + \left(\mathcal{M}^{(1)}+o(1)\right)A\left(F^{\leftarrow}\left(1 - (1-q)^k\right)\right) + \left( \mathcal{N}^{(1)}+o(1)\right)(1-q)^k\right],\nonumber
\end{flalign*}
where  $\rho=(-\alpha)\vee \beta$ and $c_1$ is defined in (\ref{eq:c_1}).
\end{theorem}
\begin{remark}\label{rem3.4}
From Theorem \ref{thm:main}, the absolute risk contribution satisfies
\begin{flalign*}
\Delta \mathrm{CoHM}_{p,q}(X,Y) &=  c_1 m_{q}^{\frac{1}{\alpha}}
F^{\leftarrow}\left(1 - (1 - q)^k\right)\times\Bigr[ (1-m_{q}^{-\frac{1}{\alpha}})+ \left(\mathcal{N}^{(1)}+o(1)\right)(1-q)^k\\
& \quad+   \left( \mathcal{M}^{(1)}-\frac{m_{q}^{-\frac{1}{\alpha}}\mathcal{J}_{\alpha, \rho, k}^{(1)}}{\alpha^2}+o(1)\right)A\left(F^{\leftarrow}\left(1 - (1-q)^k\right)\right) \Bigr],
\end{flalign*}
and relative risk contribution is given by 
\begin{flalign*}
\Delta \mathrm{CoHM}^R_{p,q}(X,Y) &= m_{q}^{\frac{1}{\alpha}}
\times\Bigr[ (1-m_{q}^{-\frac{1}{\alpha}})+ \left(\mathcal{N}^{(1)}+o(1)\right)(1-q)^k\\
& \quad+   \left( \mathcal{M}^{(1)}-\frac{\mathcal{J}_{\alpha, \rho, k}^{(1)}}{\alpha^2}+o(1)\right)A\left(F^{\leftarrow}\left(1 - (1-q)^k\right)\right)\Bigr].
\end{flalign*}
The first term $c_1m_q^{1/\alpha}F^{\leftarrow}(1 - (1-q)^k)(1-m_{q}^{-1/{\alpha}})$ represents the first-order contribution, which is proportional to the quantile and scales with the factor $m_q^{1/\alpha} - 1$. This result is consistent with the results under the first-order condition. The second-order correction consists of two components: the term $(\mathcal{M}^{(1)}-m_{q}^{-{1}/{\alpha}}\mathcal{J}_{\alpha, \rho, k}^{(1)}/{\alpha^2} )A(F^{\leftarrow}(1 - (1-q)^k))$ captures the curvature of the tail through the auxiliary function $A(\cdot)$, while the term $\mathcal{N}^{(1)}(1-q)^k$ captures the interaction between the FGM dependence structure and the finite conditioning probability. For the relative risk contribution, dividing by $\mathrm{HM}_q(X)$ yields $\Delta \mathrm{CoHM}^R_{p,q}(X,Y)$ is mainly controlled by the first-order 
contribution --- $m_q^{1/\alpha} - 1$, while under the second-order conditions it is also influenced by $(\mathcal{M}^{(1)}-\mathcal{J}_{\alpha, \rho, k}^{(1)}/{\alpha^2} )A(F^{\leftarrow}(1 - (1-q)^k))$ and  $\mathcal{N}^{(1)}(1-q)^k$. Compared with the results in \ref{subsection First-order asymptotics},
the relative risk contribution depends on the FGM parameter $r$, the reference risk level $p$, and the confidence level $q$, which captures a more comprehensive and precise structure of the relative risk contribution under the second-order conditions. 
\end{remark}

\begin{theorem}\label{thm:second}
Consider the $\mathrm{CoHM}$ risk measure defined in (\ref{eq:CoHM}) for some $k\geq 1$ with a bivariate risk vector \((X,Y)\)  that satisfies the FGM dependence given by (\ref{eq:9}). Assume that \( F \in \mathrm{MDA}(\Psi_\alpha) \) for some $\alpha>0$ with \(0<\hat{x}<\infty \) and  \( P(X = \hat{x}) = 0 \). If, further, \( \overline{F}\left( \hat x-\frac{1}{ \cdot }\right) \in 2\mathrm{RV}_{
-\alpha, \beta} \) for some $\beta\leq 0$ with an auxiliary function \( A(\cdot) \), then, as $q\uparrow 1$,
\begin{flalign*}
\hat{x}-\mathrm{CoHM}_{p,q}(X|Y) &= c_2m_{q}^{-\frac{1}{\alpha}} 
\left(\hat{x}-F^{\leftarrow}\left(1 - (1 - q)^k\right)\right)\Bigl[ 1 -   \left(\mathcal{M}^{(2)}+o(1)\right)\\
& \quad\times A\left(\left(\hat{x}-F^{\leftarrow}\left(1 - (1-q)^k\right)\right)^{-1}\right) -  \left( \mathcal{N}^{(2)}+o(1)\right)(1-q)^k\Bigr],
\end{flalign*}
where  $\rho=(-\alpha)\vee \beta$ and $c_2$ is defined in (\ref{eq:c_2}).
\end{theorem}
\begin{remark}
From Theorem \ref{thm:second}, the absolute risk contribution satisfies
\begin{flalign*}
\Delta \mathrm{CoHM}_{p,q}(X,Y) &=  c_2 m_{q}^{-\frac{1}{\alpha}}
\left(\hat{x}-F^{\leftarrow}\left(1 - (1 - q)^k\right)\right)\times\Bigr[ (m_{q}^{\frac{1}{\alpha}}-1)+ \left(\mathcal{N}^{(2)}(1-q)^k+o(1)\right)\\
& \quad+   \left( \mathcal{M}^{(2)}-\frac{m_{q}^{\frac{1}{\alpha}}\mathcal{J}_{\alpha, \rho, k}^{(2)}}{\alpha^2}+o(1)\right)A\left(\left(\hat{x}-F^{\leftarrow}\left(1 - (1-q)^k\right)\right)\right) \Bigr],
\end{flalign*}
and relative risk contribution is given by 
\begin{flalign*}
\Delta \mathrm{CoHM}^R_{p,q}(X,Y) &=  m_{q}^{-\frac{1}{\alpha}}
\times\Bigr[ (m_{q}^{\frac{1}{\alpha}}-1)+ \left(\mathcal{N}^{(2)}(1-q)^k+o(1)\right)\\
& \quad+   \left( \mathcal{M}^{(2)}-\frac{\mathcal{J}_{\alpha, \rho, k}^{(2)}}{\alpha^2}+o(1)\right)A\left(\left(\hat{x}-F^{\leftarrow}\left(1 - (1-q)^k\right)\right)\right)\Bigr].
\end{flalign*}
The first term $c_2 m_{q}^{-\frac{1}{\alpha}}
\left(\hat{x}-F^{\leftarrow}\left(1 - (1 - q)^k\right)\right)(m_{q}^{\frac{1}{\alpha}}-1)$ represents the first-order contribution, which is proportional to the quantile and scales with the factor $1-m_q^{-1/\alpha}$. This result is consistent with the results under a first-order condition. The second-order correction consists of two components: the term $(\mathcal{M}^{(2)}-m_{q}^{{1}/{\alpha}}\mathcal{J}_{\alpha, \rho, k}^{(2)}/{\alpha^2} )A\left(\left(\hat{x}-F^{\leftarrow}\left(1 - (1-q)^k\right)\right)\right)$ captures the curvature of the tail through the auxiliary function $A(\cdot)$, while the term $\mathcal{N}^{(2)}(1-q)^k$ captures the interaction between the FGM dependence structure and the finite conditioning probability. 
For the relative risk contribution, we obtain  $\Delta \mathrm{CoHM}^R_{p,q}(X,Y)$ is mainly controlled by the first-order 
contribution --- $1-m_q^{-1/\alpha}$, while under the second-order conditions it is also influenced by $(\mathcal{M}^{(2)}-\mathcal{J}_{\alpha, \rho, k}^{(2)}/{\alpha^2} )A\left(\left(\hat{x}-F^{\leftarrow}\left(1 - (1-q)^k\right)\right)\right)$ and  $\mathcal{N}^{(2)}(1-q)^k$. Compared with the results in \ref{subsection First-order asymptotics},
the relative risk contribution depends on the FGM parameter $r$, the reference risk level $p$, and the confidence level $q$, which captures a more comprehensive and precise structure of the relative risk contribution under the second-order conditions. 
\end{remark}

\begin{theorem}\label{thm:third}
Consider the $\mathrm{CoHM}$ risk measure defined in (\ref{eq:CoHM}) for some $k\geq 1$ with a bivariate risk vector \((X,Y)\)  that satisfies the FGM dependence given by (\ref{eq:9}). If \( F \in \mathrm{MDA}(\Lambda)\) with  \( 0<\hat{x}\leq\infty \)  and  \( P(X = \hat{x}) = 0 \), then, as \( q \uparrow 1 \),
\begin{equation*}
\mathrm{CoHM}_{p,q}(X|Y) = F^{\leftarrow}\left(1 - (1 - q)^k\right)
+a\left( (1 - q)^{-k}\right)\left(\lambda_k+\log m_{1}\right)(1+o(1))
\end{equation*}
for $\hat{x}=\infty$, and
\begin{equation*}
\hat{x}-\mathrm{CoHM}_{p,q}(X|Y) = \hat{x}-F^{\leftarrow}\left(1 - (1 - q)^k\right)
-a\left( (1 - q)^{-k}\right)\left(\lambda_k+\log m_{1}\right)(1+o(1))
\end{equation*}
for $\hat{x}<\infty$, where $a(t)$ is the auxiliary function of $U$ due to $U\in\mathrm{ERV}_0$ and $\lambda_k=\log\Gamma(k+1)-k\log k+k$.
\end{theorem}
\begin{remark}
From Theorem \ref{thm:third}, the absolute risk contribution satisfies
\begin{flalign*}
\Delta \mathrm{CoHM}_{p,q}(X,Y) &=a\left( (1 - q)^{-k}\right)\log m_{1}(1+o(1)),
\end{flalign*}
and relative risk contribution is given by 
\begin{flalign*}
\Delta \mathrm{CoHM}^R_{p,q}(X,Y) &= \frac{a\left( (1 - q)^{-k}\right)\log m_{1}}{F^{\leftarrow}\left(1 - (1 - q)^k\right)
+a\left( (1 - q)^{-k}\right)\left(\lambda_k+\log m_{1}\right)}(1+o(1)).
\end{flalign*}
Thus, it can be easily verified that $\lim_{q\uparrow1}\Delta \mathrm{CoHM}_{p,q}^{R}(X,Y) =0$, as we can choose $a(\cdot)=E\left[X-U(\cdot)\mid X>U(\cdot)\right]$ from Theorems B.2.2 and 1.16 in \citet{{deHaan2006}}, indicating that the conditioning effect is of a smaller order than the unconditional risk measure itself. For $\hat{x} < \infty$, a similar argument with the appropriate auxiliary function yields the same qualitative conclusion.
\end{remark}

\subsection{A multivariate case}\label{MCOHM}

Consider the random vector \((X, \mathbf{Y}) = (X, Y_1, \ldots, Y_n)\) defined on a common probability space, where \(\mathbf{Y} = (Y_1, \ldots, Y_n)\) denotes the vector of reference risks. For a given vector of confidence levels \(\mathbf{p} = (p_1, \ldots, p_n) \in (0,1)^n\), let \(\mathbf{y}_{\mathbf{p}} = (G_1^{\leftarrow}(p_1), \ldots, G_n^{\leftarrow}(p_n))\), and the conditioning event \(\{\mathbf{Y} > \mathbf{y}_{\mathbf{p}}\}\) is understood componentwise, as in the definition of the $\mathrm{MCoHM}$ risk measure in Section \ref{multivariate CoHM}. To derive its second-order asymptotics under the regime \(q\uparrow1\), we equip the vector \((X,\mathbf{Y})\) with a multivariate FGM copula. Let the marginal distribution of \(X\) be \(F\), and the marginal distribution of each \(Y_i\) be \(G_i\) for \(i=1,\dots,n\). The joint distribution is assumed to be of multivariate FGM form
\begin{align}\label{eq:multi_FGM}
\Pi(x,y_1,\ldots,y_n)
&=
F(x)\prod_{i=1}^n G_i(y_i)
\Bigg[
1+\sum_{i=1}^n r_i\bigl(1-F(x)\bigr)\bigl(1-G_i(y_i)\bigr)\nonumber \\
&\quad\quad\quad\quad\quad\quad\quad\quad\quad\quad+\sum_{1\le i<j\le n}\theta_{ij}\bigl(1-G_i(y_i)\bigr)\bigl(1-G_j(y_j)\bigr)
\Bigg],
\end{align}
where $r_i,\theta_{ij}\in[-1,1]$ are subject to the usual FGM admissibility
conditions. From (\ref{eq:multi_FGM}), the conditional survival function of $X$
given $\{\mathbf{Y}>\mathbf{y}_{\mathbf{p}}\}$ equals to

\begin{align}
\overline{F}_{X|\mathbf{Y}>\mathbf{y}_{\mathbf{p}}}(t)
&=
\frac{
1+\sum_{1\le i<j\le n}\theta_{ij}G_i(y_{p_i})G_j(y_{p_j})
+\sum_{i=1}^n r_i G_i(y_{p_i})
}{
1+\sum_{1\le i<j\le n}\theta_{ij}G_i(y_{p_i})G_j(y_{p_j})
}\,\overline{F}(t) \notag\\
&\quad\times
\left[
1-
\frac{\sum_{i=1}^n r_i G_i(y_{p_i})}{
1+\sum_{1\le i<j\le n}\theta_{ij}G_i(y_{p_i})G_j(y_{p_j})
+\sum_{i=1}^n r_i G_i(y_{p_i})
}\,\overline{F}(t)
\right]. \label{eq:multi_cond_surv}
\end{align}

Define the aggregate dependence parameter

\begin{equation*}\label{eq:mathcal_G_def}
\mathcal{G}_{\mathbf{r},\boldsymbol{\theta},\mathbf{p}}
:=
\frac{\sum_{i=1}^n r_i G_i(y_{p_i})}{
1+\sum_{1\le i<j\le n}\theta_{ij}G_i(y_{p_i})G_j(y_{p_j})}.
\end{equation*}
The following lemma shows that the conditional survival function inherits the
same second‑order regular variation structure as the marginal $\overline F$,
with $\mathcal{G}$ playing the role of the multivariate dependence term.

\begin{lemma}\label{lem:multi_2RV}
Let $(X,\mathbf{Y})$ follow the multivariate FGM distribution
(\ref{eq:multi_FGM}). Suppose $\overline{F}\in 2\mathrm{RV}_{-\alpha,\beta}$
with $\alpha>0$, $\beta\leq 0$ and an auxiliary function $A(\cdot)\in\mathrm{RV}_\beta$.
Then

\begin{equation*}\label{eq:multi_2RV_result}
\overline{F}_{X|\mathbf{Y}>\mathbf{y}_{\mathbf{p}}}
\in 2\mathrm{RV}_{-\alpha,\rho}
\end{equation*}
with $ \rho=(-\alpha)\vee\beta$ and an auxiliary function

\begin{equation}\label{eq:multi_aux}
B(t)=A(t)+\frac{\alpha\,\mathcal{G}_{\mathbf{r},\boldsymbol{\theta},\mathbf{p}}}
{1+\mathcal{G}_{\mathbf{r},\boldsymbol{\theta},\mathbf{p}}}\,\overline{F}(t).
\end{equation}

\end{lemma}

\begin{proof}
For any $x>0$, write $\mathcal{G}:=\mathcal{G}_{\mathbf{r},\boldsymbol{\theta},\mathbf{p}}$.
From (\ref{eq:multi_cond_surv}) and following a similar argument in the proof of Lemma A.5, we derive

\begin{align*}
\frac{\overline{F}_{X|\mathbf{Y}>\mathbf{y}_{\mathbf{p}}}(tx)}
{\overline{F}_{X|\mathbf{Y}>\mathbf{y}_{\mathbf{p}}}(t)}
&=
x^{-\alpha}
\Bigg[
1+A(t)\frac{x^\beta-1}{\beta}
+\frac{\alpha\mathcal{G}}{1+\mathcal{G}}\frac{x^{-\alpha}-1}{-\alpha}\overline{F}(t)
+o\left(A(t)+\overline{F}(t)\right)
\Bigg].
\end{align*}
This ends the proof.
\end{proof}

The auxiliary function $B(\cdot)$ in (\ref{eq:multi_aux}) contains a pure marginal
term $A(\cdot)$ and a dependence-induced term proportional to $\overline{F}(\cdot)$.
The parameter $\rho=(-\alpha)\vee\beta$ picks the slower of the two decay
rates, ensuring that the expansion remains valid. Consequently, the entire
second-order asymptotic machinery developed in Theorems
\ref{thm:main}--\ref{thm:third} applies to the multi-conditional setting
under the substitution of $ rG(y_p)$ by $ \mathcal{G}_{{\bf r},{\bf \theta},{\bf p}}.$

For instance, in the Fréchet case, the expansion reads
\begin{align*}
\mathrm{MCoHM}_{\mathbf{p},q}(X|\mathbf{Y})
&=
c_1 m_q^{1/\alpha}
F^{\leftarrow}\bigl(1-(1-q)^k\bigr) \notag\\
&\quad\times
\Bigg[
1+
\left(\mathcal{M}^{(1)}+o(1)\right)
A\left(F^{\leftarrow}\bigl(1-(1-q)^k\bigr)\right)
+
\left(\mathcal{N}^{(1)}+o(1)\right)(1-q)^k
\Bigg], \label{eq:MCoHM_Frechet}
\end{align*}
where
$m_q=1+\mathcal{G}_{\mathbf{r},\boldsymbol{\theta},\mathbf{p}}
(1-(1-q)^k)$, $m_1=1+\mathcal{G}_{\mathbf{r},\boldsymbol{\theta},\mathbf{p}}$,
and $\mathcal{M}^{(1)}$, $\mathcal{N}^{(1)}$ are the symbols defined in Section \ref{subsec:3.1} with $rG(y_p)$ replaced by $\mathcal{G}$. The Weibull and Gumbel analogs are obtained by the same substitution.

\begin{remark}
The aggregate parameter $\mathcal{G}$ summarizes the tail dependence. Its
denominator contains the pairwise dependence $\theta_{ij}$ among the reference
risks; when these are positively dependent, conditioning on all of them
simultaneously produces a smaller incremental effect than the sum of the
individual terms $r_iG_i(y_{p_i})$. This is consistent with the
“diminishing marginal” effect of correlated conditioning events. However, the FGM
copula restricts the analysis to weak dependence; extensions to
strong tail dependence would require different second‑order tools and are left
for future work.
\end{remark}

\subsection{A further discussion on joint asymptotics} \label{Joint_asy}
In this subsection, we further discuss the iterated limit problems of the asymptotics of the $\mathrm{CoHM}_{p,q}[X| Y]$ risk measure as $p\uparrow 1$ and $q\uparrow 1$. It can approximate the $\mathrm{CoHM}$ risk measure when the risks $X$ and $Y$ are both exposed at higher confidence levels that reflect the economic/insurance reality. 

Our analysis has focused on the regime $q\uparrow1$ with $p$ fixed. A
complementary perspective, studied by \citet{TXZ2026}, takes $p\uparrow1$ with $q$ fixed and reveals a proportional relationship between
$\mathrm{CoHM}_{p,q}(X|Y)$ and the unconditional HM risk measure of an
auxiliary variable under bivariate regular variation. A natural question is
whether the two sequential limits are equivalent:
\[
L_1:=\lim_{q\uparrow1}\lim_{p\uparrow1}\mathrm{CoHM}_{p,q}(X|Y),
\qquad
L_2:=\lim_{p\uparrow1}\lim_{q\uparrow1}\mathrm{CoHM}_{p,q}(X|Y).
\]

Suppose that the reference risk $Y$ itself belongs to an $\mathrm{MDA}$ (Fréchet, Weibull or Gumbel) and, in the heavy-tailed cases, satisfies the corresponding
second-order regular variation condition. Then, for each fixed $q$,
$\lim_{p\uparrow1}F_{X|Y>y_p}$ exists as a proper distribution
$\widetilde F_q$, satisfying monotonicity, regularity, right-continuity, and the conditional measure inherits the $\mathrm{2RV}$ properties of $\overline F$ with an updated auxiliary function. Applying the results of
Theorems \ref{thm:main}--\ref{thm:third} to $\widetilde F_q$ yields $L_1$
explicitly.

Conversely, we can also take the limit as $p \uparrow1$ in Theorems \ref{thm:main}-\ref{thm:third}, which leads to the iterated limit $L_2.$ Concretely, for $G \in \mathrm{MDA}(\Phi_{\alpha^{*}})$ and  $\overline{G} \in \mathrm{2RV}_{-\alpha^{*},\beta^{*}}$ for some $\alpha^{*}>0$ and $\beta^{*}<0$, applying Lemma \ref{lem:first} below and absorbing negligible terms, we can derive the new asymptotic results by replacing $G(y_p)$ with $p$ in Theorems \ref{thm:main}-\ref{thm:third}. For $G \in \mathrm{MDA}(\Psi_{\alpha^{*}})$ as well as $\overline{G}\left(\hat{x}-\frac{1}{\cdot}\right) \in  \mathrm{2RV}_{-\alpha^{*},\beta^{*}}$ for some $\alpha^{*}>0$ and $\beta^{*}<0$ or $G \in \mathrm{MDA}(\Lambda)$ that implies its tail quantile function $U_G \in \mathrm{ERV}$,  one can easily obtain the corresponding asymptotic results by replacing $G(y_p)$ with $p$ in Theorems \ref{thm:main}-\ref{thm:third}. A direct
comparison of the resulting expressions shows that
\begin{equation}\label{eq:exchangeability}
L_1 = L_2
\end{equation}
whenever the second-order parameters and the tail indices of $X$ and $Y$
satisfy the natural dominance conditions. In the Gumbel case, the equality
follows from the slowly varying nature of the auxiliary function $a(\cdot)$.

\begin{remark}
The identity (\ref{eq:exchangeability}) is not automatic; it is based on the
assumption that the limiting conditional distribution $\widetilde F$ (as
$p\uparrow1$) remains in the same $\mathrm{2RV}$ class as $\overline F$. For the FGM
copula, this holds 0conditioning only alters the constant
$m_1=1+r\lim_{p\uparrow1}G(y_p)$, without changing the tail index or the
second-order parameter. For copulas with tail dependence, $\widetilde F$ could
have a heavier or lighter tail than $\overline F$, and the sequential limits
may differ.
\end{remark}

It is important to emphasize that the joint double limit
\[
\lim_{(p,q)\to(1,1)}\mathrm{CoHM}_{p,q}(X|Y)
\]
cannot be studied within the current second-order regular variation framework
because the simultaneous divergence of two arguments destroys the uniform
convergence properties on which our expansions rely. A full treatment would
require a genuinely two-dimensional extreme-value theory for conditional
quantiles, which is beyond the scope of this paper. For practical purposes,
the iterated limit $L_1$ (or $L_2$) already provides a theoretically
justified asymptotic benchmark when both the primary loss and the reference risk approach their extreme domains.

\section{Numerical illustrations}\label{Numerical_ill}

In this section, three numerical examples along with some simulation studies are first performed to validate the precision of the second-order asymptotics of the $\mathrm{CoHM}$ coherent risk measure presented in Theorems \ref{thm:main}-\ref{thm:third} via the crude Monte Carlo method. 
The primary risk distributions in these three examples come from three different MDAs, respectively. We use optimize in R to minimize the objective function in (\ref{eq:CoHM}) over \(x\), and then compute the corresponding value to obtain the exact numerical approximation of the $\mathrm{CoHM}$ risk measure. We aim to demonstrate the improvement of our second-order asymptotic formulae over the first-order ones. 
\begin{example} \label{ex:example1}
(The Fr\'echet case) Let the primary risk \(X\) be a r.v. with Lomax distribution function \(F\) given by  
\begin{equation} \label{eq:pareto}
F(x) = 1 - \left( \frac{\theta}{\theta + x} \right)^{\alpha}, \quad x > 0, 
\end{equation}
where \(\alpha > 0\) and \(\theta > 0\). Moreover, we assume that the primary risk $X$ and the reference risk $Y$ are dependent according to a FGM dependence with parameter \(r\).

In this case, it can be easily verified that \(\overline{F} \in 2\mathrm{RV}_{-\alpha,-1} \) with an auxiliary function \(A(t)= \alpha\theta/t\), and therefore \(F \in \mathrm{MDA}(\Phi_\alpha)\).  

In this numerical illustration, we set \(k = 2\), \(p = 0.80\) (or $G(y_p)=0.8$), \(r = 0.60\), \(\alpha = 5\) and \(\theta = 6\). We compute the exact value of \(\mathrm{CoHM}_{p,q}(X|Y)\) together with the first-order asymptotic expansion using the same method as in \citet{LiuYi2025}, and simulate the second‑order asymptotic expansion by Theorem \ref{thm:main} in our paper. The exact value, the first-order and the second-order asymptotic of \(\mathrm{CoHM}_{p,q}(X|Y)\) as a function of the confidence level \(q \in [0.95,\,0.999]\) are plotted in Figure \ref{fig:example}. 
\begin{figure}[htbp]
    \centering
    \begin{subfigure}{0.48\textwidth}
        \centering
        \includegraphics[width=\textwidth]{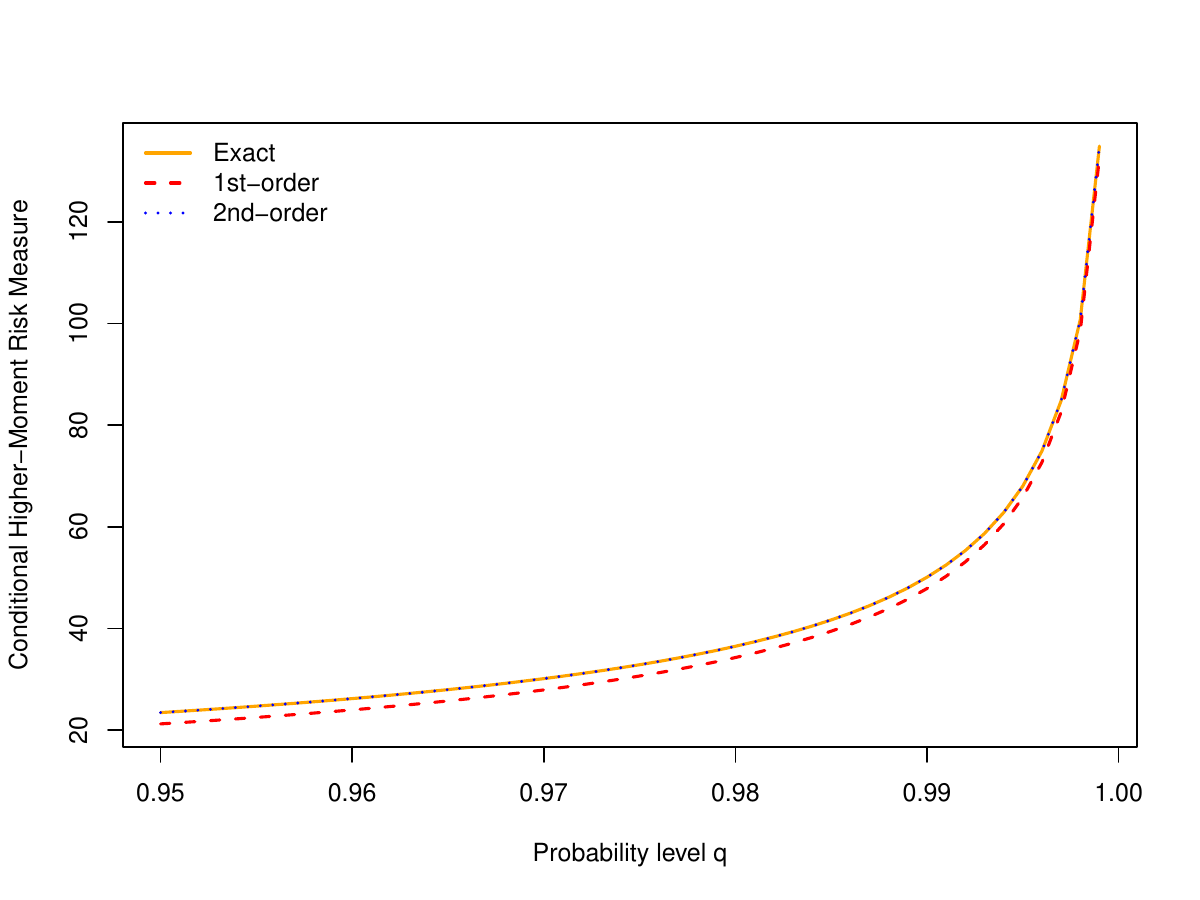}
    \end{subfigure}
    \hfill
    \begin{subfigure}{0.48\textwidth}
        \centering
        \includegraphics[width=\textwidth]{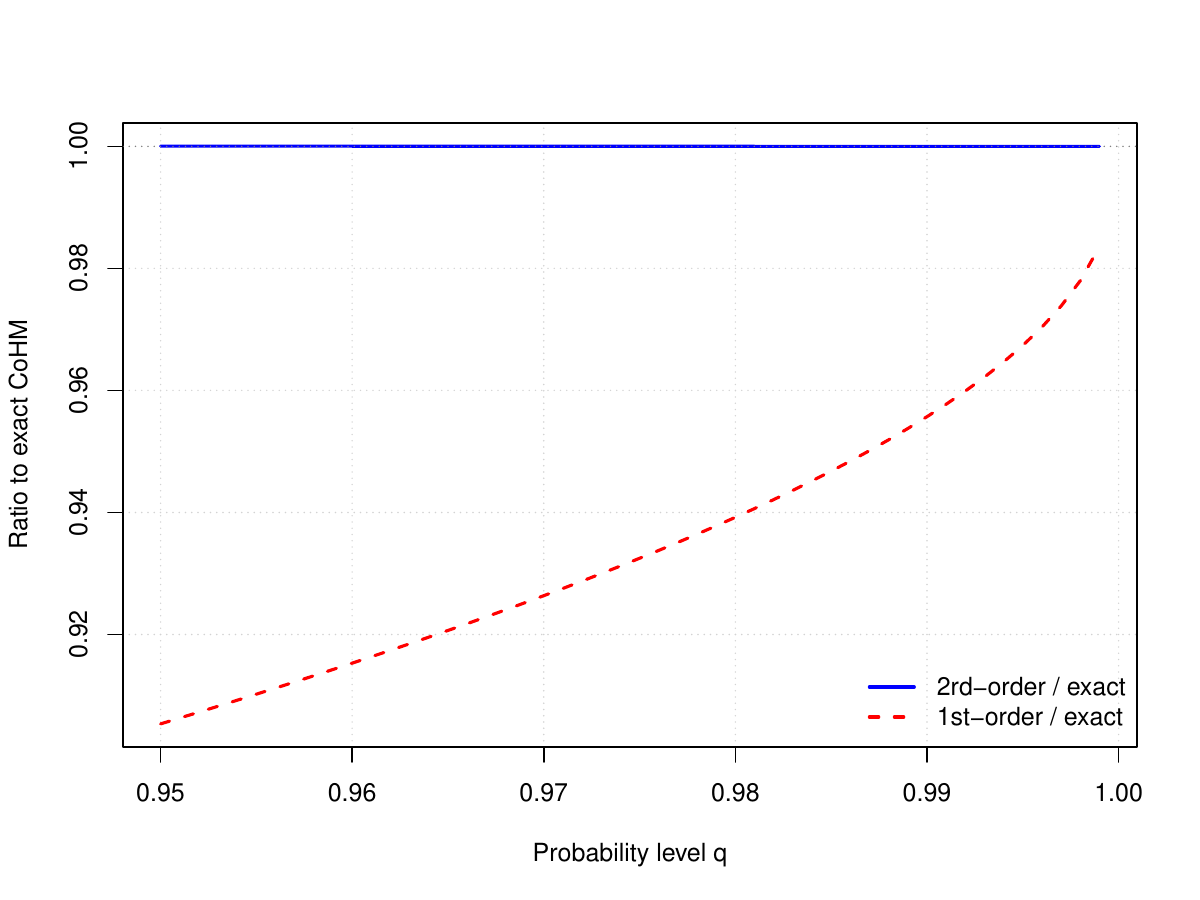}
    \end{subfigure}
    \caption{Comparison of the first-/second-order asymptotic values with the exact values of
    \(\mathrm{CoHM}_{p,q}(X|Y)\) in the presence of the Fr\'echet case.}
    \label{fig:example}
\end{figure}

As shown in Figure \ref{fig:example}, in general, both first- and second-order asymptotic effects are basically good in this case. However, the second-order asymptotic is slightly closer to the exact value, demonstrating that the second‑order asymptotic provides a more accurate approximation than the first‑order one.
\end{example}

\begin{example}
(The Weibull case)  Let the primary risk \(X\) be a r.v. following a Beta distribution with density  
\[
f(x) = \frac{x^{a-1}(1-x)^{b-1}}{{B}(a,b)}, \qquad 0 < x < 1,\; a,b > 0.
\]
Moreover, we assume that the primary risk $X$ and the reference risk $Y$ are dependent according to a FGM dependence with parameter \(r\). 

It is well known that  \(F \in \mathrm{MDA}(\Psi_b)\) with upper endpoint \(\hat{x} = 1\). By Example 3.3.17 in \citet{Embrechts1997}, it holds that
\[
\overline{F}\Bigl(1-\frac{1}{t}\Bigr) = \int_{1-1/t}^{1} f(x)\,dx = \int_{t}^{\infty} f\Bigl(1-\frac{1}{s}\Bigr)\frac{1}{s^{2}}\,ds.
\]
Meanwhile, it is easy to see that, as \(t\to\infty\),
\[
f\Bigl(1-\frac{1}{t}\Bigr)\frac{1}{t^2} = \frac{t^{-(b+1)}}{B(a,b)}\Bigl(1-\frac{a-1}{t}+o\Bigl(\frac{1}{t}\Bigr)\Bigr), 
\]
which implies $$f\Bigl(1-\frac{1}{t}\Bigr)\frac{1}{t^2}\in 2\mathrm{RV}_{-b-1,-1}$$ with an auxiliary function \({(a-1)}/{t}\). Therefore, it follows from Proposition 6 in \citet{HuaJoe2011} that, as \(t\to\infty\),
\[
\overline{F}\Bigl(1-\frac{1}{t}\Bigr) = \frac{t^{-b}}{bB(a,b)}\Bigl(1-\frac{b(a-1)}{b+1}\,t+o\Bigl(\frac{1}{t}\Bigr)\Bigr),\qquad t\to\infty.
\]
which leads to $$\overline{F}(1-1/t)\in 2\mathrm{RV}_{-b,-1}$$ with an auxiliary function \({(ab-b)}/{(tb+t)}\).
\begin{figure}[htbp]
    \centering
    \begin{subfigure}{0.48\textwidth}
        \centering
        \includegraphics[width=\textwidth]{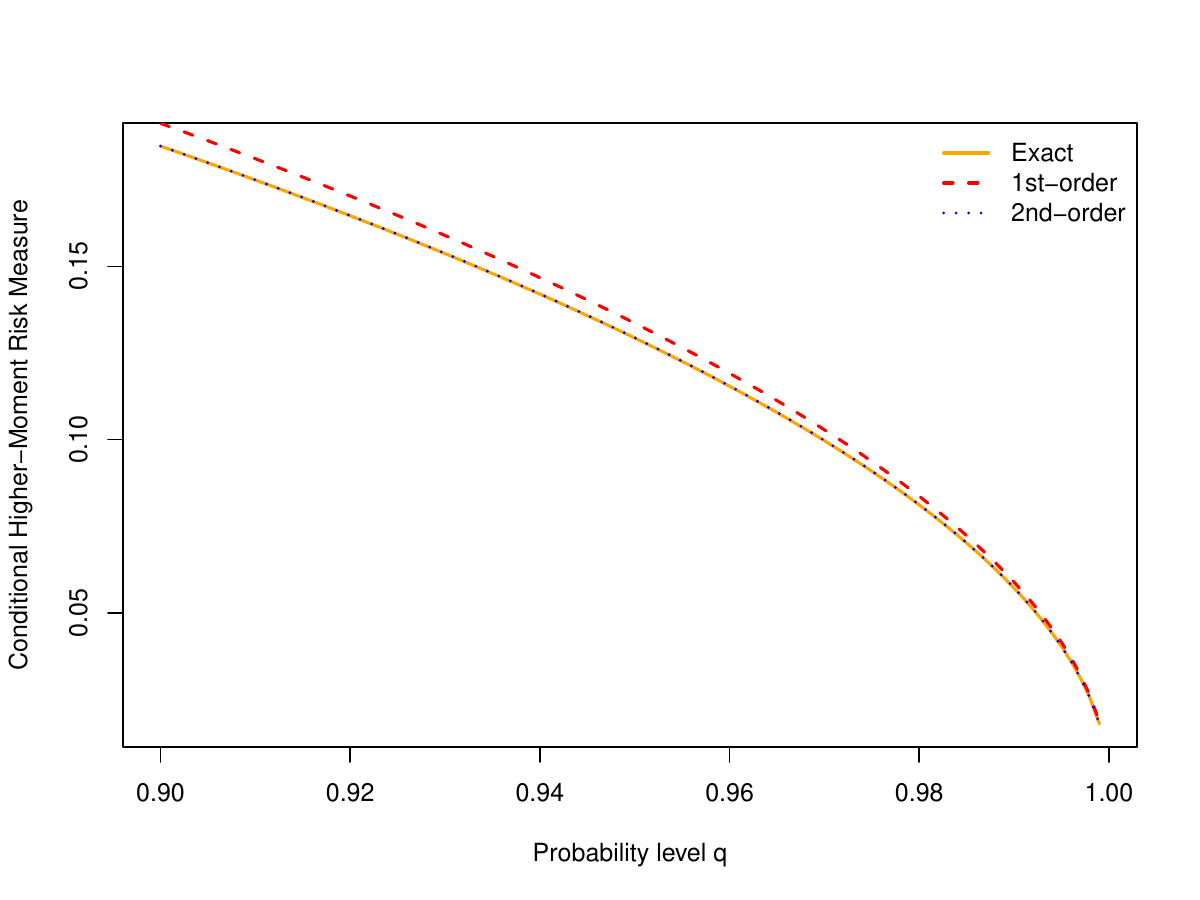}
    \end{subfigure}
    \hfill
    \begin{subfigure}{0.48\textwidth}
        \centering
        \includegraphics[width=\textwidth]{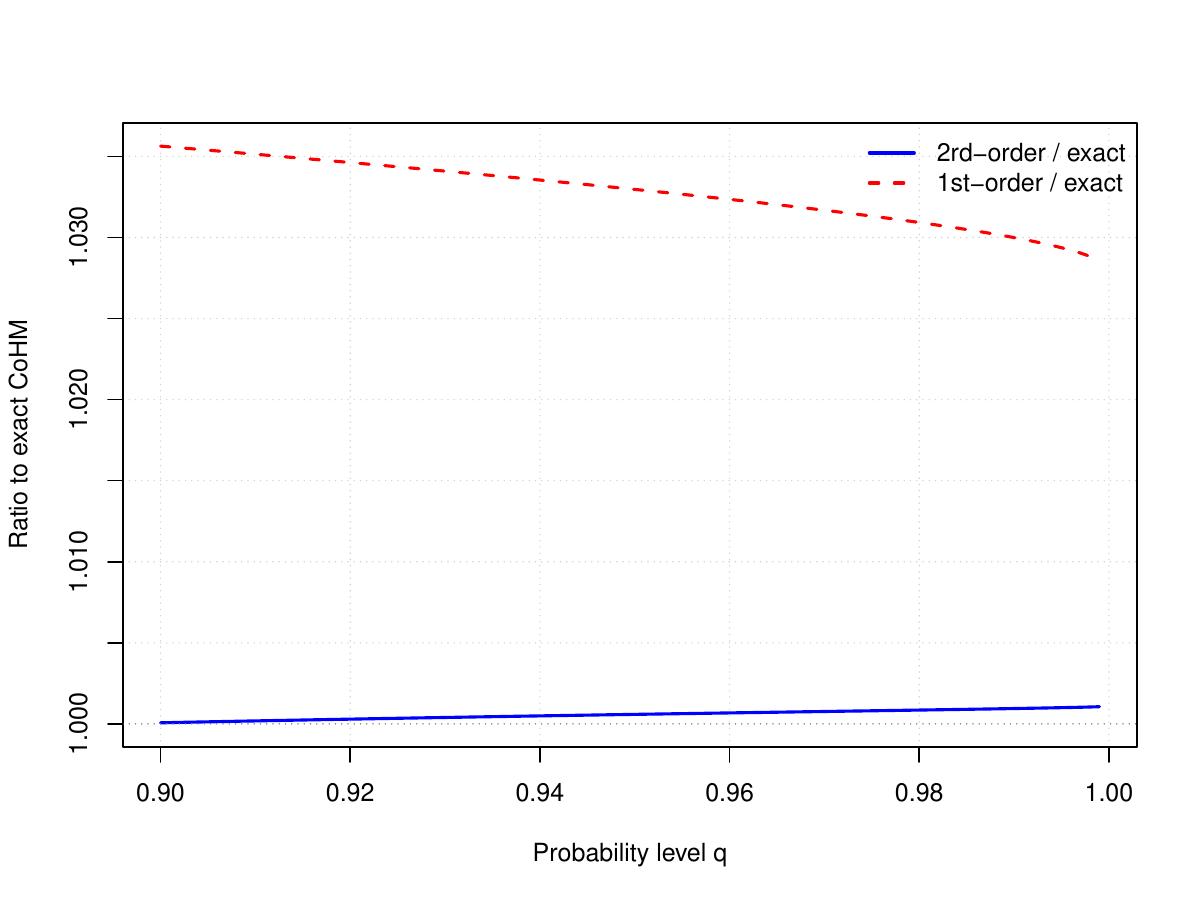}
    \end{subfigure}
    \caption{Comparison of the first-/second-order asymptotic values with the exact values of
    \(\mathrm{CoHM}_{p,q}(X|Y)\) in the presence of the Weibull case.}
    \label{fig:2}
\end{figure}

In this numerical illustration, we set \(k = 3\), \(r = 0.2\), \(d = 2\), \(b = 6\) and \(p=0.9\).  Similarly to Example \ref{ex:example1},  we also compute the exact value of \(\mathrm{CoHM}_{p,q}(X|Y)\) together with the first-order asymptotic expansion using the same method as in \citet{LiuYi2025}, and simulate the second‑order asymptotic expansion by Theorem \ref{thm:second} in our paper. The exact value, the first-order and the second-order asymptotic of \(\mathrm{CoHM}_{p,q}(X|Y)\) as a function of the confidence level \(q \in [0.90,\,0.999]\) are shown in Figure \ref{fig:2}.

It is evident from Figure \ref{fig:2} that there is a certain deviation between the first-order asymptotic and the exact value, while the second-order asymptotic is very close to the exact value, which demonstrates that the second‑order expansion provides a more accurate approximation than the first‑order one, confirming the theoretical improvement.
\end{example}

\begin{example}
(The Gumbel case) Let the primary risk \(X\) be a r.v. following a lognormal distribution with distribution function  
\[
F(x) = \Phi\left(\frac{\log x - \mu}{\sigma}\right),\qquad x>0,\;\mu\in\mathbb{R},\;\sigma>0,
\]
where \(\Phi\) denotes the standard normal distribution function. Moreover, we assume that the primary risk $X$ and the reference risk $Y$ are dependent according to a FGM dependence with parameter \(r\). 

\begin{figure}[htbp]
    \centering
    \begin{subfigure}{0.48\textwidth}
        \centering
        \includegraphics[width=\textwidth]{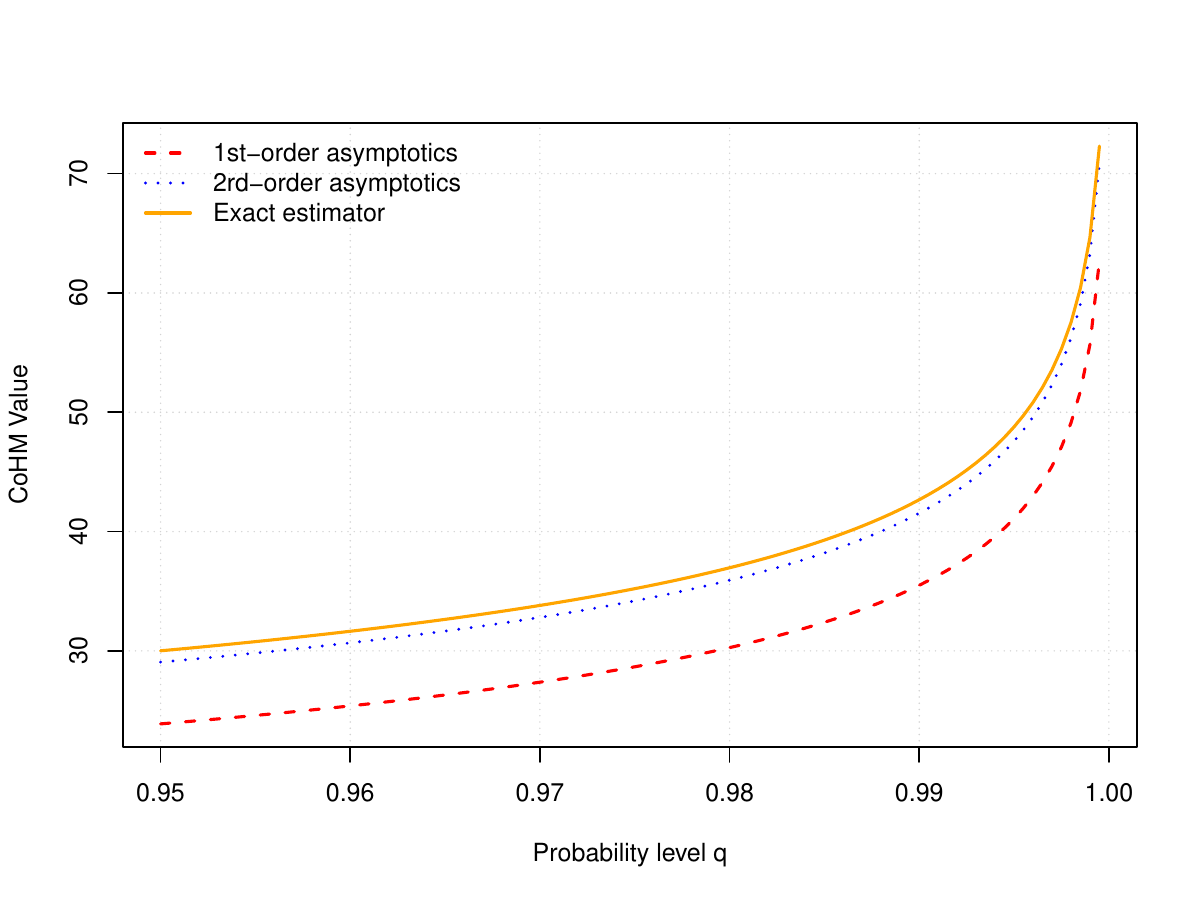}
    \end{subfigure}
    \hfill
    \begin{subfigure}{0.48\textwidth}
        \centering
        \includegraphics[width=\textwidth]{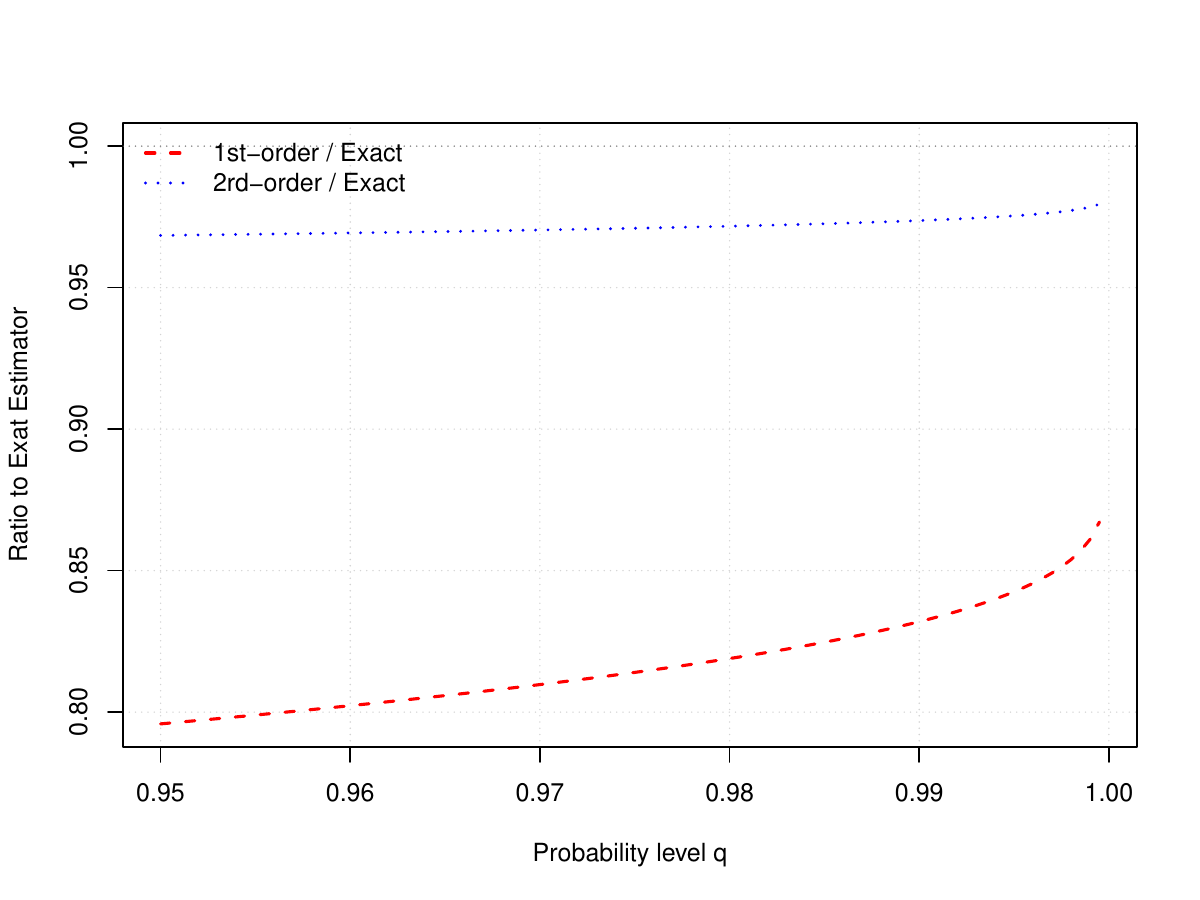}
    \end{subfigure}
    \caption{Comparison of the first-/second-order asymptotic values with the exact values of
    \(\mathrm{CoHM}_{p,q}(X|Y)\) in the presence of the Gumbel case.}
    \label{fig:3}
\end{figure}

Due to Example 3.3.35 in \citet{Embrechts1997}, we know that \(F \in \mathrm{MDA}(\Lambda)\) with $\hat{x}=\infty$ and an auxiliary function 
\[
a(x) = \frac{\Big(1-\Phi\big((\log x-\mu)/\sigma\big)\Big)\,\sigma x}{\phi\big((\log x-\mu)/\sigma\big)},\qquad x>0,
\]
where \(\phi\) denotes the standard normal density.

In this numerical illustration, we set \(k=1.5\), \(\mu=2\), \(\sigma=0.5\), \(p=0.9\) and \(r=0.2\).  We compute the exact value of \(\mathrm{CoHM}_{p,q}(X|Y)\) together with the first-order asymptotic expansion using the same method as in \citet{LiuYi2025}, and simulate the second‑order asymptotic expansion by Theorem \ref{thm:third} in our paper. The exact value, the first-order and the second-order asymptotic of \(\mathrm{CoHM}_{p,q}(X|Y)\) as a function of the confidence level \(q \in [0.95,\,0.999]\) are shown in Figure \ref{fig:3}. 

Figure \ref{fig:3} clearly indicates that there is still a significant deviation between the first-order asymptotic and the exact value, but the second-order asymptotic is basically consistent with the exact value, which shows that the second‑order expansion provides a substantially more accurate approximation than the first‑order one, confirming the theoretical improvement.

\end{example}

\section{Empirical analysis}\label{Empirical_ana}
We now evaluate the practical accuracy of the proposed second-order asymptotic expansions using two real-world insurance datasets, covering both light-tailed and heavy-tailed loss processes.
\subsection{A light-tailed case}
The dataset used in this light-tailed empirical analysis is \texttt{dataCar} from the R package \texttt{insuranceData}, which contains 67856 vehicle insurance claim records of an Australian insurance company during 2004--2005, including two vehicle body types, namely the station wagon (STNWG) and the
sedan (SEDAN); see \citet{deJong2008} for more statistical analysis about this dataset. We follow the data filtering and preprocessing procedure of \citet{WangLi2026}: select samples with vehicle body types $\mathrm{SEDAN}(X)$ and $\mathrm{STNWG}(Y)$, remove zero claims, and convert the claim amounts to thousand dollars. Meanwhile, a bootstrap procedure is employed to generate a large sample of size \(10^6\), the marginal parameters of the Gamma distributions for SEDAN and STNWG are estimated by the method of moments, and the dependence parameter \(r\) of the FGM dependence structure is obtained by maximum likelihood estimation (see Table \ref{tab:2}). For detailed descriptive statistics, goodness-of-fit tests, and estimation results, we refer to Section 5 of \citet{WangLi2026}. Adopting the same estimation methods, we further calculate the empirical and asymptotic estimates of the $\mathrm{CoHM}$ risk measure, and compare the approximation performance of the first-order and second-order asymptotic results.
\begin{table}[htbp]
\centering
\caption{Parameter estimation results for vehicle insurance}
\label{tab:2}
\begin{tabular*}{\textwidth}{@{\extracolsep{\fill}} lccc @{}}
\toprule
Parameter & \(X\) & \(Y\) \\
\midrule
Shape \(\hat{\alpha}\) & 0.3840 & 0.2440 \\
Scale \(\hat{\beta}\) & 0.2112 & 0.1211 \\
\midrule
FGM parameter & \(\hat{r}\) & \(0.0028\) \\
\bottomrule
\end{tabular*}
\end{table}

As done in \citet{WangLi2026}, we also choose the gamma distribution as the reference type of distributions for $X$ and $Y$ to match the lighted-tailed nature of the loss data. A r.v. \(X\) is said to follow a Gamma distribution, denoted by \(X \sim \text{Gamma}(\alpha, \beta)\), if its probability density function is given by
\[
f(x) = \frac{\beta^{\alpha}}{\Gamma(\alpha)} x^{\alpha-1} e^{-\beta x},\qquad x > 0,
\]
where \(\alpha > 0\) is the shape parameter, \(\beta > 0\) is the scale parameter. It is easy to see the tail probability satisfies
\[
\overline{F}(x) = P(X > x) \sim \frac{\beta^{\alpha-1}}{\Gamma(\alpha)} x^{\alpha-1} e^{-\beta x},\quad x \to \infty.
\]
Hence, the Gamma distribution belongs to MDA(\(\Lambda\)) in extreme value theory with an auxiliary function \(a(\cdot) \to 1/\beta\) by Remark 1.2.7 of \citet{deHaan2006}. Suppose that the marginal distributions of  $\mathrm{SEDAN}(X)$ and $\mathrm{STNWG}(Y)$ both follow Gamma distributions, whose shape parameter \(\alpha\) and scale parameter \(\beta\) are estimated by the method of moments. The dependence structure between \(X\) and \(Y\) is characterized by the FGM dependence structure. Based on the pseudo-observations from the bootstrap sample, we derived the estimation of FGM parameter \(r\) via maximum likelihood estimation. 

Next, we compute the first-order and second-order asymptotic $\mathrm{CoHM}$ estimates by Theorem \ref{thm:third}.
For the empirical estimate, we take the same way as in the heavy-tailed case.
\begin{table}[htbp]
\centering
\small
\renewcommand{\arraystretch}{0.8}
\caption{Comparison of the empirical and asymptotic $\mathrm{CoHM}$ estimates($p=0.9$) for light-tailed case}
\label{tab:cohm_comparison}
\begin{tabular*}{\textwidth}{@{}l@{\extracolsep{\fill}}cccccc@{}}
\toprule
$k$ & $q$ & Emp & First & Err$_{1}$ (\%) & Second & Err$_{2}$ (\%) \\
\midrule
\multirow{4}{*}{1.0} & 0.95 & 11.87 & 7.66 & 35.46 & 12.39 & 4.38 \\
                     & 0.97 & 13.92 & 9.59 & 31.07 & 14.32 & 2.90 \\
                     & 0.99 & 18.44 & 13.94 & 24.40 & 18.67 & 1.24 \\
                     & 0.995 & 21.35 & 16.78 & 21.43 & 21.50 & 0.70 \\
\midrule
\multirow{4}{*}{1.4} & 0.95 & 18.50 & 12.29 & 33.60 & 17.70 & 4.32 \\
                     & 0.97 & 21.51 & 15.18 & 29.46 & 20.59 & 4.28 \\
                     & 0.99 & 28.11 & 21.59 & 23.19 & 27.01 & 3.92 \\
                     & 0.995 & 32.34 & 25.74 & 20.41 & 31.16 & 3.66 \\
\midrule
\multirow{4}{*}{1.6} & 0.95 & 22.03 & 14.70 & 33.25 & 20.40 & 7.38 \\
                     & 0.97 & 25.52 & 18.07 & 29.18 & 23.77 & 6.85 \\
                     & 0.99 & 33.16 & 25.53 & 23.02 & 31.23 & 5.83 \\
                     & 0.995 & 38.04 & 30.33 & 20.27 & 36.03 & 5.29 \\
\bottomrule
\end{tabular*}
\end{table}
The results in Table \ref{tab:cohm_comparison} show that the second-order asymptotics of $\mathrm{CoHM}$ approximate the true tail moment risk measure more accurately, especially at high confidence levels ($q \ge 0.99$) for $X$. Although first-order asymptotics are theoretically simpler, they exhibit large deviations with finite samples and real data, requiring a second-order correction to meet practical accuracy requirements. This validates the effectiveness of our proposed second-order asymptotic theory. In the light-tailed case, under the high confidence level $p$ of $Y$ ($p>0.9$), $\mathrm{CoHM}$ capture the tail risk capital that should be reserved for the sedan business when the STNWG business suffers a rare large claim. Providing a clear numerical basis for independent yet scenario-augmented capital allocation between the two vehicle insurance lines.

\subsection{A heavy-tailed case}
In this subsection, we illustrate the asymptotic results of the $\mathrm{CoHM}$ risk measure through another real‑data application. The dataset used in this heavy-tailed empirical analysis is the Danish fire loss dataset, available as \texttt{danishmulti} in the \texttt{R} package \texttt{CASdatasets}, consisting of $N = 2167$ fire losses from 1980 to 1990; each split into $2$ coverage types: building losses and contents losses. See \citet{McNeil1997} for more information about this dataset. 
\begin{table}[htbp]
\centering
\caption{Descriptive statistics of building and content losses after removing zeros}
\label{tab:descriptive}
\setlength{\tabcolsep}{12pt}   
\begin{tabular}{lrrrrrr}
\hline
Variable & Min. & $Q_1$ & Median & Mean & $Q_3$ & Max. \\
\hline
Building & 0.02319 & 0.88 & 1.22 & 1.87 & 1.93 & 95.17 \\
Content & 0.00083 & 0.27 & 0.48 & 1.63 & 1.19 & 132.01 \\
\hline
\end{tabular}
\end{table}

In the original dataset, some fire incidents resulted in only one type of loss (i.e., either building loss or content loss was zero). Since this study focuses on the extreme risk when both loss types occur simultaneously, and zero values would interfere with the fitting of continuous marginal distributions and the subsequent probability integral transformation, we removed observation pairs where either building loss or content loss was zero. After filtering, 1502 valid observations remained, representing approximately $ 69.3\% $ of the total sample.
Table \ref{tab:descriptive} presents the basic descriptive statistics of the building and content losses after removing zeros. Both variables exhibit clear positive skewness and heavy-tailed behavior; in particular, the median of content loss is much smaller than its mean, indicating substantial probability mass in the tail. 
\begin{table}[htbp]
\centering
\caption{Parameter estimation results for Danish fire loss}
\label{tab:fire}
\begin{tabular*}{\textwidth}{@{\extracolsep{\fill}} lccc @{}}
\toprule
Parameter & \(X\) & \(Y\) \\
\midrule
Shape \(\hat{\alpha}\) & 9.77 & 1.76 \\
Scale \(\hat{\theta}\) & 16.05 & 1.16 \\
\midrule
FGM parameter & \(\hat{r}\) & \(0.744\) \\
\bottomrule
\end{tabular*}
\end{table}

To capture the heavy-tailed nature of the loss data, we fit Lomax distributions to both margins. The survival function of this distribution is given in (\ref{eq:pareto}). Using maximum likelihood estimation, the resulting parameter estimates are presented in Table \ref{tab:fire}. The estimated shape parameter for building loss is relatively large, indicating a faster tail decay and thus a moderately heavy tail; by contrast, content loss exhibits a heavier tail with a higher probability of extreme losses. To assess the goodness of fit, we plot the CDF comparison plots and QQ plots of the marginal distributions (see Figure \ref{fig:QQ}). From the plots, both variables align well with the Lomax distribution in both the central and tail regions, with only minor deviations in the extreme upper tail of the building loss. Overall, the Lomax model adequately captures the tail behavior of the loss distributions.
\begin{figure}[htbp]
    \centering
    \includegraphics[width=0.8\textwidth]{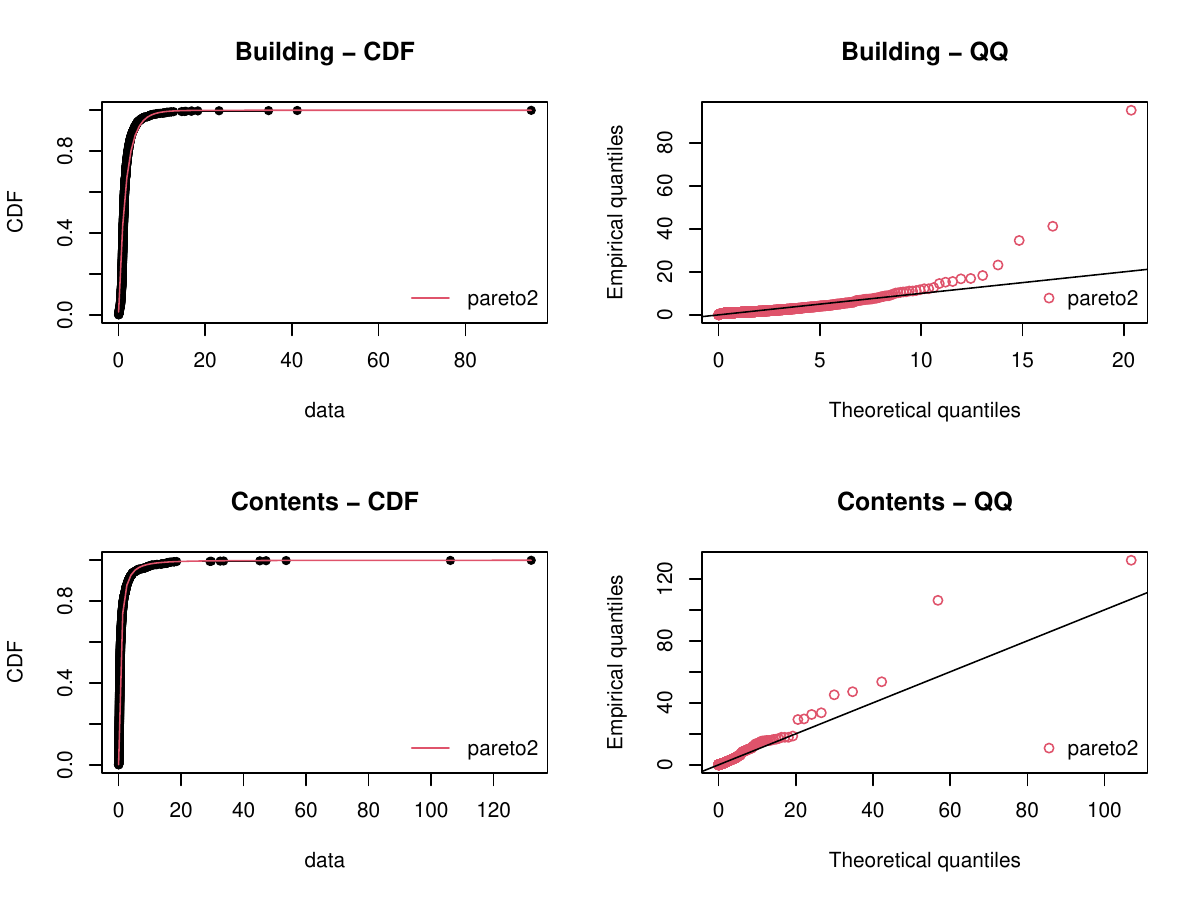}
    \caption{Q-Q plots for Building and Contents losses with fitted Lomax distributions}
    \label{fig:QQ}
\end{figure}
After determining the marginal distributions, we apply the probability integral transformation to convert the two loss series into pseudo‑observations on $\left[0,1\right]$, and then model the bivariate dependence structure. We adopt the FGM dependence structure to describe the weak contagion effect between building and content losses. Figure \ref{fig:FGM} displays the actual pseudo‑observations along with points simulated from the fitted FGM dependence structure; the scatter patterns are broadly consistent, further supporting the appropriateness of the assumed dependence structure.
\begin{figure}[htbp]
    \centering
    \includegraphics[width=0.7\textwidth]{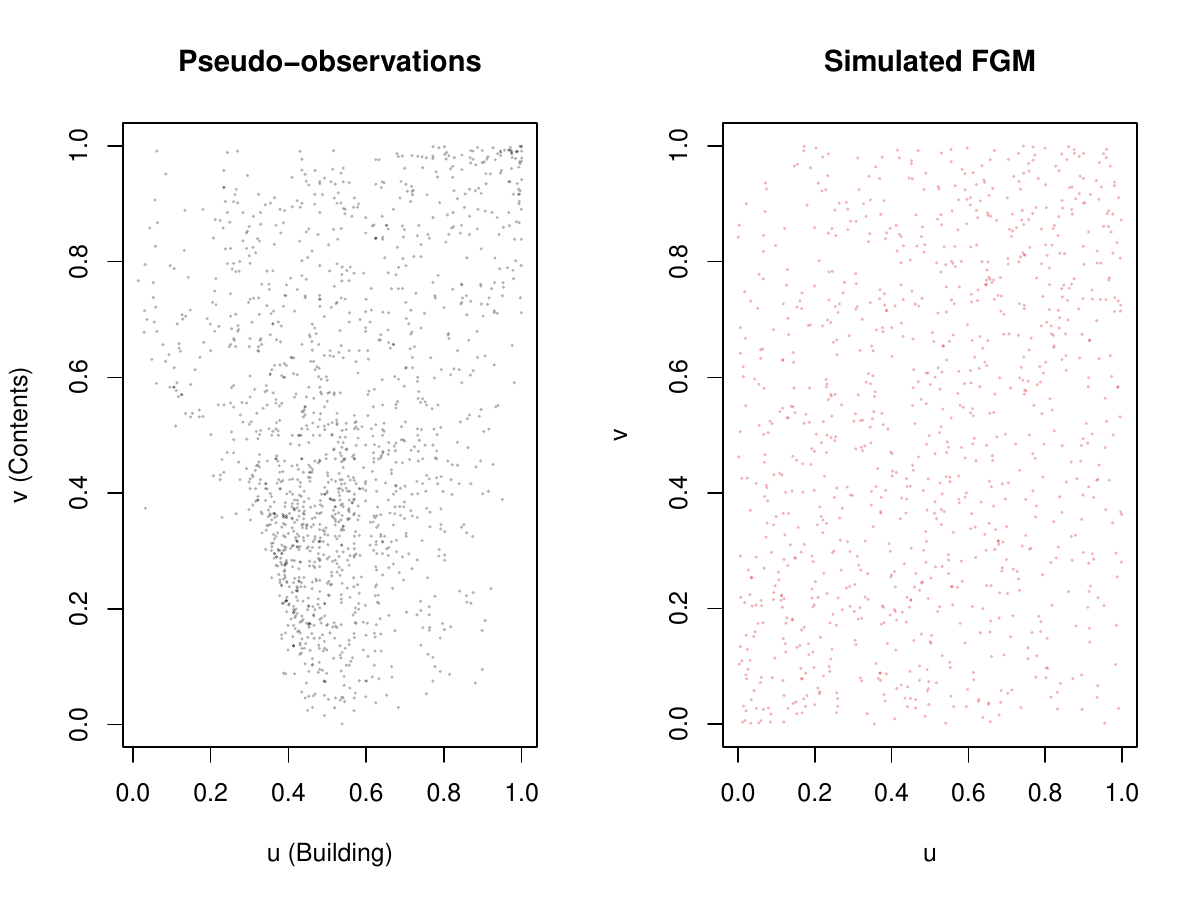}
    \caption{Pseudo-observations from data (left) and simulated FGM dependence structure (right)}
    \label{fig:FGM}
\end{figure}

Under the fully parametric $\mathrm{Lomax}$–$\mathrm{FGM}$ model framework described above, we employ three approaches to obtain estimates of $\mathrm{CoHM}$ in order to achieve stable and accurate risk measures in finite samples and at extreme confidence levels. On the basis of the estimated marginal parameters and the FGM copula parameter, we generate a large simulated dataset of size $10^6$. Based on the simulated sample, for given \(p, q, k\), we first compute the sample quantile \(y_p\) of \(Y\) and select the subsample satisfying \(Y > y_p\), whose size is denoted by \(n_p\). Then we define the objective function
\[
h(x) = x + \frac{1}{1-q} \left( \frac{1}{n_p} \sum_{i: Y_i > y_p} (X_i - x)_+^k \right)^{1/k},
\]
and find its minimizer on the interval \([0, x_{\mathrm{upper}}]\) via one-dimensional optimization (e.g., using the \texttt{optimize} function of \texttt{R}), where \(x_{\mathrm{upper}}\) is taken as the quantile of the conditional sample \(X \mid Y > y_p\). The resulting minimum value is the empirical $\mathrm{CoHM}$ estimator. Furthermore, we adopt an extrapolation estimator proposed by \citet{LiuYi2025} to address the issue of sparse extreme observations at very high confidence levels. This method effectively circumvents numerical obstacles while maintaining theoretical consistency.
Based on Theorem \ref{thm:main} and Example \ref{ex:example1}, we directly computed the first-order and second-order theoretical asymptotic values of $\mathrm{CoHM}$. In the computation, three typical cases of the higher-order moment order \(k = 1, 1.5, 2\) are considered, with the confidence level \(q\) ranging from 0.95 to 0.999. From Table \ref{tab:cohm_comparison for heavy} we can see second-order asymptotic outperforms the first-order one and agrees with the extrapolated estimates over most intervals, confirming the practical utility of the asymptotic formula.
\begin{figure}[htbp]
    \centering
    \begin{subfigure}[b]{0.45\textwidth}
        \includegraphics[width=\textwidth]{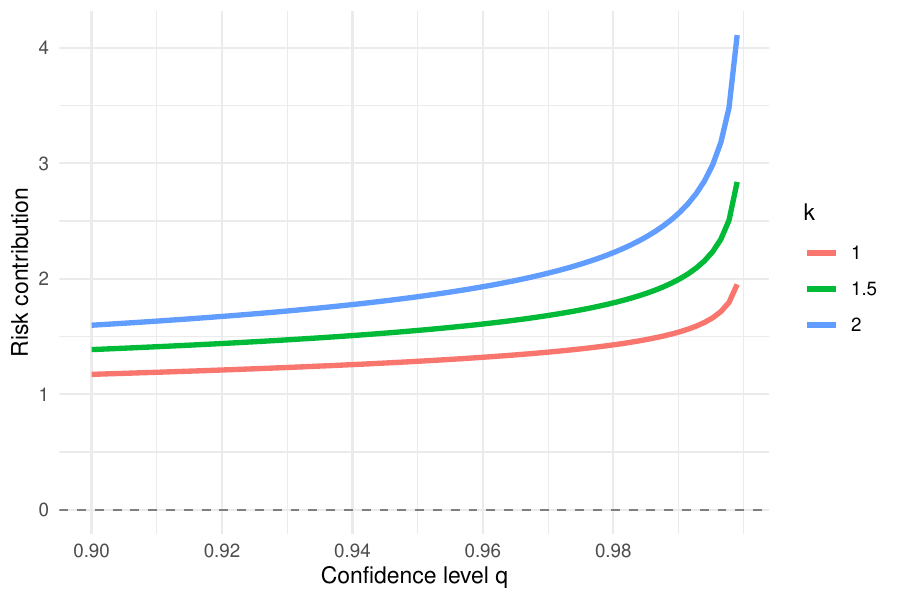}
        \caption{$p=0.9$}
        \label{fig:risk_q}
    \end{subfigure}
    \hfill
    \begin{subfigure}[b]{0.45\textwidth}
        \includegraphics[width=\textwidth]{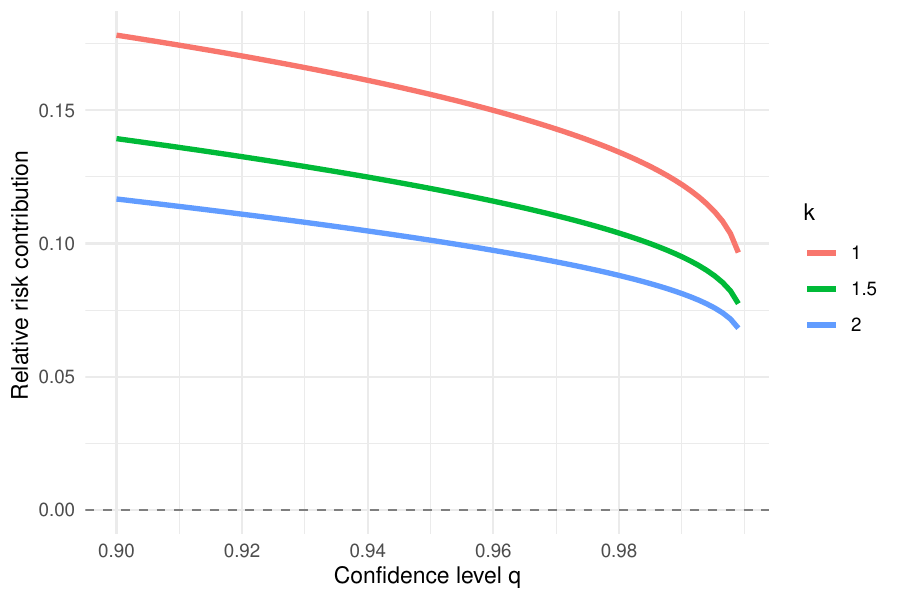}
        \caption{$p=0.9$}
        \label{fig:rel_q}
    \end{subfigure}
    \par\medskip
    \begin{subfigure}[b]{0.45\textwidth}
        \includegraphics[width=\textwidth]{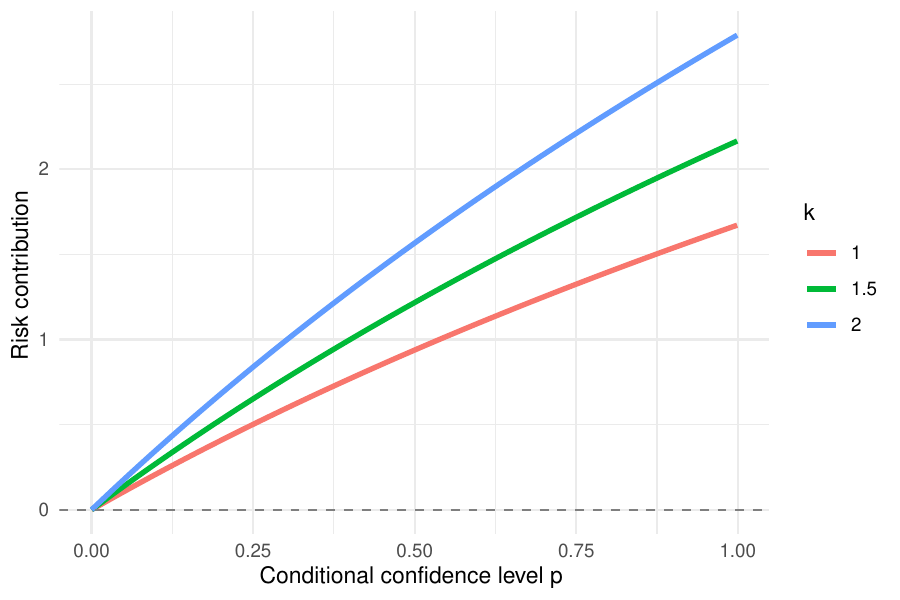}
        \caption{$q=0.99$}
        \label{fig:risk_p}
    \end{subfigure}
    \hfill
    \begin{subfigure}[b]{0.45\textwidth}
        \includegraphics[width=\textwidth]{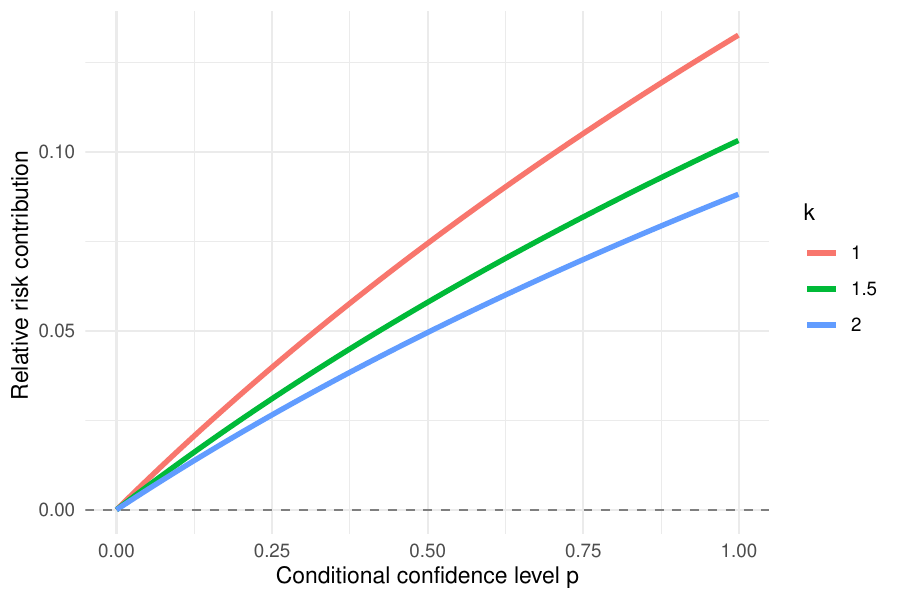}
        \caption{$q=0.99$}
        \label{fig:rel_p}
    \end{subfigure}
    \caption{Risk contribution and relative risk contribution under the second-order $\mathrm{CoHM}$: variations with $q$ and $p$ for selected $k$}
    \label{fig:combined}
\end{figure}

To grasp the insurance meaning of $\mathrm{CoHM}$ risk measure, we consider two configurations: Setting A, in which building loss is taken as the primary risk variable $X$ and content loss as the conditioning variable $Y$ (denoted as Building $\leftarrow$ Contents), and Setting B, in which content loss acts as the primary risk variable $X$ while building loss serves as the conditioning variable $Y$ (Contents $\leftarrow$ Building). From Table \ref{tab:cohm_compare for p} we obtain in Setting A, the $\mathrm{CoHM}$ exhibits a relatively flat slope with respect to the confidence level $p$. This indicates that, even when content losses become extremely severe, the additional increase in extreme building losses remains limited. Such a decoupling suggests that building structures inherently possess an effective upper bound on losses during extreme fires; consequently, insurers may adopt relatively conservative capital reserves for building risk even under worsening contents‑loss scenarios.
In contrast, Setting B renders the $\mathrm{CoHM}$ considerably more sensitive to $p$, causing it to rise more sharply. Once building damage reaches an extreme level, the contents are almost inevitably destroyed, necessitating a substantial upward revision of risk capital for contents losses. 
\begin{table}[htbp]
\centering
\small
\renewcommand{\arraystretch}{0.8}
\caption{Comparison of the empirical and asymptotic $\mathrm{CoHM}$ estimates($p=0.9$) for heavy-tailed case}
\label{tab:cohm_comparison for heavy}
\begin{tabular*}{\textwidth}{@{}l@{\extracolsep{\fill}}cccccc@{}}
\toprule
$k$ & $q$ & Emp & First & Err$_{1}$ (\%) & Second & Err$_{2}$ (\%) \\
\midrule
\multirow{4}{*}{1.0} & 0.95 & 9.60 & 7.69 & 19.87 & 9.53 & 0.71 \\
                     & 0.97 & 10.97 & 9.08 & 17.29 & 10.91 & 0.56 \\
                     & 0.99 & 14.19 & 12.30 & 13.32 & 14.13 & 0.41 \\
                     & 0.995 & 16.41 & 14.52 & 11.52 & 16.35 & 0.36 \\
\midrule
\multirow{4}{*}{1.5} & 0.95 & 14.63 & 12.21 & 16.55 & 14.42 & 1.41 \\
                     & 0.97 & 17.13 & 14.70 & 14.20 & 16.91 & 1.28 \\
                     & 0.99 & 23.23 & 20.76 & 10.63 & 22.97 & 1.11 \\
                     & 0.995 & 27.63 & 25.14 & 9.04 & 27.35 & 1.04 \\
\midrule
\multirow{4}{*}{2.0} & 0.95 & 20.58 & 17.55 & 14.70 & 20.07 & 2.47 \\
                     & 0.97 & 24.61 & 21.53 & 12.51 & 24.05 & 2.29 \\
                     & 0.99 & 34.87 & 31.64 & 9.24 & 34.16 & 2.02 \\
                     & 0.995 & 42.62 & 39.30 & 7.81 & 41.81 & 1.90 \\
\bottomrule
\end{tabular*}
\end{table}

\begin{table}[htbp]
\centering
\caption{Dependence of second-order CoHM on the conditioning comfidence level \(p\) with \(q=0.99\) and \(k=1.5\).}
\label{tab:cohm_compare for p}
\begin{tabular}{lccccccc}
\toprule
\multicolumn{1}{c}{$p$} & 0.00 & 0.18 & 0.36 & 0.54 & 0.73 & 0.91 & 0.99 \\
\midrule
$\mathrm{CoHM}$(Setting A) & 20.97 & 21.46 & 21.89 & 22.29 & 22.65 & 22.98 & 23.14 \\
$\mathrm{CoHM}$(Setting B) & 216.95 & 233.19 & 248.62 & 263.35 & 277.47 & 291.07 & 297.69 \\
\bottomrule
\end{tabular}
\end{table}
As shown in Figure \ref{fig:combined}, the risk contribution is monotone increasing in confidence level $q$ and relative confidence level $p$ and the relative risk contribution is  concentrated within the interval $\left[0,0.2\right]$. And this result is consistent with that in section \ref{rem3.4}. From the practical meaning of risk contribution and relative risk contribution we obtain that neglecting the dependence (assuming  the primary risk variable $X$ and the conditioning variable $Y$  are independent) would lead to an underestimation of the overall risk, particularly at an extreme confidence level, which highlights the importance of incorporating content loss into building loss assessment.

These contrasting sensitivities provide a foundation for differentiated pricing: for properties with high fire‑resistant construction, the parameters of Setting A can be invoked to justify conservative pricing, whereas for properties with a high proportion of flammable contents, the extreme scenarios captured by Setting B warrant closer attention and higher capital allocations.

\section{Conclusion}\label{Conclusion}
This paper systematically investigates the second-order asymptotic theory of the $\mathrm{CoHM}$ coherent risk measure under the FGM copula that models weak tail contagion between primary loss and reference risk factors. From a theoretical perspective, we adopt EVT and second-order regular variation techniques to derive explicit refined asymptotic expansions for $\mathrm{CoHM}$ when the confidence level $q\uparrow1$, covering three canonical MDAs for marginal loss distributions.
The derived second-order expansions integrate tail curvature parameters, FGM weak contagion intensity and finite quantile bias terms, which capture the joint impact of marginal tail characteristics and cross-risk dependence. In addition, we formally decompose absolute and relative systemic risk contributions $\Delta \mathrm{CoHM}$ and $\Delta\mathrm{CoHM^R}$ into leading first-order components and two separate second-order correction terms, respectively reflecting tail curvature effects and finite-probability interaction of FGM contagion. Finally,  we further generalize the bivariate $\mathrm{CoHM}$ to the multi-conditional $\mathrm{MCoHM}$ risk measure accommodating multiple systemic reference risks, and demonstrate that its second-order asymptotic structure only requires a modified dependence parameter without substantial re-derivation of core formulae.


In terms of simulation and empirical studies, numerical Monte Carlo experiments based on Lomax (Fréchet), Beta (Weibull) and lognormal (Gumbel) distributions consistently validate the superiority of our second-order formulae. Across all tail types and extreme confidence levels, second-order asymptotics drastically cut estimation errors relative to conventional first-order approximations, especially for high quantiles where first-order expansions exhibit severe under/overestimation biases. Two real insurance datasets are utilized for empirical validation: light-tailed Australian motor insurance claims following Gamma distribution and heavy-tailed Danish fire loss data fitted by Lomax margins. Parametric estimation confirms that the FGM copula adequately captures weak tail contagion between distinct insurance business lines. Comparative error analysis reveals that second-order $\mathrm{CoHM}$ estimators yield substantially lower relative errors at high quantiles, whereas first-order approximations exhibit severe under/overestimation biases. Moreover, empirical $\mathrm{CoHM}$ indicators quantify heterogeneous systemic spillover effects across loss categories, providing quantitative evidence for differentiated capital reserve allocation and cross-line reinsurance pricing in insurance risk management practice.

From practical risk management perspectives, this work provides an accurate, tractable analytical approximation for conditional tail risk under weak systemic contagion, avoiding costly large-scale Monte Carlo simulation when evaluating extreme loss scenarios for insurers and financial institutions. For regulatory purposes, the scale-invariant relative risk contribution index offers a stable benchmark to quantify the systemic vulnerability of individual risk exposures in the midst of market-wide catastrophes. This research bears clear limitations. The entire theoretical framework is built upon the FGM copula, which only captures weak cross-risk contagion and fails to characterize strong asymptotic dependence common in severe financial crises. Future research directions include extending the second-order asymptotic methodology to general bivariate and multivariate copulas that allow strong tail dependence, relaxing the $q\uparrow 1$ asymptotic regime to simultaneously analyze $(p,q)\rightarrow (1,1)$ joint limits, and developing consistent statistical estimators for the second-order tail and dependence parameters to facilitate real-time systemic risk monitoring.

\appendix
\section{Appendix}\label{Appendix}
\subsection{Some lemmas}
In this sequel, we prepare some lemmas necessary for the proofs of the main results. Lemma \ref{lem:first+0} is attributed to \cite{TangYang2012}. Lemma \ref{lem:first} is due to Exercise 2.11 of \citet{deHaan2006}.  Lemma \ref{lem:first+1} can be easily derived by an inequalities in \cite{TangYang2014}.
Lemma \ref{lem:fifth} is a slight modification of Theorems 4.1 and 4.5 in \cite{MaoHu2012}, which are crucial to the proof of the main results.
Lemmas \ref{lem:fourth}-\ref{lem:eighth} verify the second-order regularly varying properties of the conditional distribution of $X$  conditioned on $\{Y>y_p\}$ with $y_p=\mathrm{VaR}_p(Y)$ when the bivariate risk vector \((X,Y)\)  satisfies the FGM dependence given by (\ref{eq:9}). 

In addition,  in what follows, we denote by \(F_{X|Y>y_p}\) the conditional distribution of $X$  conditioning on $\{Y>y_p\}$ with $y_p=\mathrm{VaR}_p(Y)$ for any $p\in(0,1)$,  and by \(U_{X|Y>y_p}\) the tail quantile function of $\overline{F}_{X|Y>y_p}$ by convention, that is, 
\begin{equation}
U_{X|Y>y_p}(t)=\left(\frac{1}{\overline F_{X|Y>y_p}}\right)^{\leftarrow}\left(t\right)=F_{X|Y>y_p}^{\leftarrow}\left(1-\frac{1}{t}\right).\nonumber
\end{equation}

\begin{lemma}\label{lem:first+0}
Let $F$ on $\mathbb{R}$ belong to the MDA of a non-degenerate distribution function, then
\begin{equation}
\overline{F}(F^{\leftarrow}(q)) \sim 1-q .\nonumber
\end{equation}
\end{lemma}

\begin{lemma}\label{lem:first}
If $\overline F\in\mathrm{2RV}_{-\alpha,\beta}$ for some $\alpha>0$ and $\beta\leq0$ with an auxiliary function $A(\cdot)$, then
\begin{equation}
{F}(F^{\leftarrow}(q)) = q\left(1+o(A(F^{\leftarrow}(q)))\right),\nonumber
\end{equation}
and 
\begin{equation}
{U}\left(1/\overline F(t)\right) = t\left(1+o(A(t))\right).\nonumber
\end{equation}
\end{lemma}

\begin{lemma}\label{lem:first+1}
If $U\in\mathrm{ERV}_{\alpha}$ for some $\alpha\in\mathbb{R}$ with a first-order auxiliary function $a(\cdot)$, then
\begin{equation}
U\left(t\pm0\right)=U(t)+o(a(t)).\nonumber
\end{equation}
\end{lemma}
\begin{proof}The proof follows directly from Lemma 2.3 in \cite{TangYang2014}.
\end{proof}

\begin{lemma}\label{lem:fifth}
Let  $X$  be a risk variable distributed by $F$ with $( P(X = \hat{x}) = 0 )$. Denote by $ U(t) $ the tail quantile function of $F$ as defined in (\ref{eq:quantile}). Let $ \phi(t) = t^k $ for some $ k \geq 1 $.

\begin{enumerate}
\item[(i)] (The Fr\'echet case) If $ U \in 2\mathrm{RV}_{\gamma, \rho} $ for some $ 0 < \gamma < \frac{1}{k} $  and $ \rho \leq 0 $ with an auxiliary function $ A(\cdot) $, then, as $ q \uparrow 1 $,
\[
\mathrm{HG}_q[X] = d_1 F^{\leftarrow}(q) \left[ 1 + \widetilde{\mathcal{J}}_{\gamma, \rho, k}^{(1)} A\left(\frac{1}{1-q}\right)(1+o(1))\right]
\]
with \[ d_1 = \frac{k^\gamma (1 - k\gamma)^{k\gamma - 1}}{(k\gamma)^{k\gamma}} \left( B\left(\frac{1}{\gamma} - k, k\right)\right)^\gamma \] 
and
\[
\widetilde{\mathcal{J}}_{\gamma, \rho, k}^{(1)} = \frac{1}{\rho} \left[ \left( \frac{1-k\gamma}{k\gamma} \right)^{k\rho} k^{\rho} \left(B\left(\frac{1}{\gamma} - k, k\right)\right)^{\rho-1} B\left(\frac{1-\rho}{\gamma} - k, k\right) - 1 \right].
\]
\item[(ii)] (The Weibull case) If $ \hat{x} - U \in 2\mathrm{RV}_{\gamma, \rho} $ for some $ \gamma < 0 $ and $ \rho \leq 0 $ with an auxiliary function $ A(\cdot) $, then, as $ q \uparrow 1 $,
\[
\hat{x} - \mathrm{HG}_q[X] = d_2 (\hat{x} - F^{\leftarrow}(q)) \times \left[ 1 + \widetilde{\mathcal{J}}_{\gamma, \rho, k}^{(2)} A\left(\frac{1}{1-q}\right)(1+o(1))\right]
\]
with 
\[
d_2 = \frac{(1 - k\gamma)^{(k-1)\gamma - 1}}{k^{(k-1)\gamma}(-\gamma)^{k\gamma}} \left( B\left(-\frac{1}{\gamma}, k\right) \right)^{\gamma}
\] 
and
\[
\widetilde{\mathcal{J}}_{\gamma, \rho, k}^{(2)} = \frac{1}{\rho} \left[ \left( \frac{1-k\gamma}{-k\gamma} \right)^{(k-1)\rho} (-\gamma)^{-\rho} \frac{(1-\rho)(1-k\gamma)}{1-\rho-k\gamma} \left(B\left(\frac{-1}{\gamma} , k\right)\right)^{\rho-1} B\left(\frac{\rho-1}{\gamma} , k\right) - 1 \right].
\]
\item[(iii)](The Gumbel case)
If \( U \in \mathrm{ERV}_0 \) with an auxiliary function \( a(\cdot) \), then, as \( q \uparrow 1 \),

\[
\mathrm{HG}_q[X] = F^{\leftarrow}(q) + a\left(\frac{1}{1-q}\right) \lambda_k \cdot (1 + o(1))
\]
for \( \hat{x} = \infty \), and
\[
\hat{x} - \mathrm{HG}_q[X] = \hat{x} - F^{\leftarrow}(q) - a\left(\frac{1}{1-q}\right) \lambda_k \cdot (1 + o(1))
\]
for \( \hat{x} < \infty \), where $\lambda_k$ is defined in Theorem \ref{thm:third}.
\end{enumerate}
\end{lemma}
\begin{lemma}\label{lem:fourth}
 Let \((X,Y)\) be a bivariate risk vector satisfying the FGM dependence given by (\ref{eq:9}) with maginal distributions $F$ and $G$, respectively. Suppose that  $\overline{F} \in 2\mathrm{RV}_{-\alpha, \beta}$ for some \(\alpha > 0 \) and \( \beta \leq 0 \) with an auxiliary function \( A(\cdot)\in \mathrm{RV}_\beta \). Then
$$\overline{F}_{X|Y>y_p} \in\mathrm{2RV}_{-\alpha,\rho}$$
 with an auxiliary function $$B(\cdot)=A(\cdot)+\frac{\alpha rG(y_p)}{1+rG(y_p)}\overline{F}(\cdot)$$ and $\rho=(-\alpha)\vee \beta$.

\begin{proof} By (\ref{eq:9}), one can easily see that
\begin{equation}\label{eq:A1}
\overline{F}_{X|Y>y_p}\left(t\right)=\overline F(t)(1+rG(y_p)F(t))=\left(1+rG(y_p)\right)\overline F(t)\left[1-\frac{rG(y_p)}{1+rG(y_p)}\overline F(t)\right].
\end{equation}
Thus, for any $x>0$, due to $\overline{F} \in 2\mathrm{RV}_{-\alpha, \beta}$, by Taylor expansion, it holds that
\begin{align*}
\frac{\overline{F}_{X|Y>y_p}\left(tx\right)}
{\overline{F}_{X|Y>y_p}\left(t\right)}
&= \frac{\overline{F}\left(tx\right)}
{\overline{F}\left(t\right)}\left[1-\frac{rG(y_p)}{1+rG(y_p)}\overline F(tx)\right]\left[1-\frac{rG(y_p)}{1+rG(y_p)}\overline F(t)\right]^{-1} \\
&= \left(x^{-\alpha} + A(t)x^{-\alpha}\frac{x^\beta-1}\beta 
(1+o(1))\right)
\left(1-\frac{rG(y_p)}{1+rG(y_p)}\overline F(tx)\right)\\
&\quad\times
\left(1+\frac{rG(y_p)}{1+rG(y_p)}\overline F(t)(1+o(1))\right)\\
&= x^{-\alpha}\left(1 + A(t)\frac{x^\beta-1}\beta (1+o(1))\right)
\left(1-\frac{rG(y_p)}{1+rG(y_p)}
\left(\overline{F}\left(tx\right)-\overline{F}\left(t\right)\right)\left(1+o(1)\right)\right)\\
&= x^{-\alpha}\left( 1+ A(t)\frac{x^\beta-1}\beta (1+o(1))\right)
\left(1+ \frac{\alpha rG(y_p)}{1+rG(y_p)}
\frac{x^{-\alpha}-1}{-\alpha}\overline{F}\left(t\right)\left(1+o(1)\right)\right)\\
&= x^{-\alpha}\left(1+A(t)\frac{x^\beta-1}\beta\left(1+o(1)\right)
+ \frac{\alpha rG(y_p)}{1+rG(y_p)}
\frac{x^{-\alpha}-1}{-\alpha}\overline{F}\left(t\right)\left(1+o(1)\right)
\right),
\end{align*}
which leads to the desired result, and this ends the proof.
\end{proof}
\end{lemma}

\begin{lemma}\label{lem:WEIBULL}
Let \((X,Y)\) be a bivariate risk vector satisfying the FGM dependence given by (\ref{eq:9}) with maginal distributions $F$ and $G$, respectively. Suppose that  $\overline{F}\left(\hat{x}-\frac{1}{\cdot}\right) \in 2\mathrm{RV}_{-\alpha, \beta}$ for some \(\alpha > 0 \) and \( \beta \leq 0 \) with an auxiliary function \( A(\cdot)\in \mathrm{RV}_\beta \). Then
$$\overline{F}_{X|Y>y_p}\left(\hat{x}-\frac{1}{\cdot}\right) \in\mathrm{2RV}_{-\alpha,\rho}$$
with an auxiliary function $$B(\cdot)=A(\cdot)+\frac{\alpha rG(y_p)}{1+rG(y_p)}\overline{F}\left(\hat{x}-\frac{1}{\cdot}\right)$$ and $\rho=(-\alpha)\vee \beta$.

\begin{proof} By (\ref{eq:9}), one can easily see that
\begin{align*}
\frac{\overline{F}_{X|Y>y_p}\left(\hat{x}-\frac{1}{tx}\right)}
{\overline{F}_{X|Y>y_p}\left(\hat{x}-\frac{1}{t}\right)}
&= \frac{\overline{F}\left(\hat{x}-\frac{1}{tx}\right)}
{\overline{F}\left(\hat{x}-\frac{1}{t}\right)}\left[1-\frac{rG(y_p)}{1+rG(y_p)}\overline F\left(\hat{x}-\frac{1}{tx}\right)\right]\left[1-\frac{rG(y_p)}{1+rG(y_p)}\overline F\left(\hat{x}-\frac{1}{t}\right)\right]^{-1} \\
&= \left(x^{-\alpha} + A(t)x^{-\alpha}\frac{x^\beta-1}\beta 
(1+o(1))\right)
\left(1-\frac{rG(y_p)}{1+rG(y_p)}\overline F\left(\hat{x}-\frac{1}{tx}\right)\right)\\
&\quad\times
\left(1+\frac{rG(y_p)}{1+rG(y_p)}\overline F\left(\hat{x}-\frac{1}{t}\right)(1+o(1))\right)\\
&= x^{-\alpha}\left(1 + A(t)\frac{x^\beta-1}\beta (1+o(1))\right)\\
&\quad\times\left(1-\frac{rG(y_p)}{1+rG(y_p)}
\left(\overline{F}\left(\hat{x}-\frac{1}{tx}\right)-\overline{F}\left(\hat{x}-\frac{1}{t}\right)\right)\left(1+o(1)\right)\right)\\
&= x^{-\alpha}\left( 1+ A(t)\frac{x^\beta-1}\beta (1+o(1))\right)\\
&\quad\times\left(1+ \frac{\alpha rG(y_p)}{1+rG(y_p)}
\frac{x^{-\alpha}-1}{-\alpha}\overline{F}\left(\hat{x}-\frac{1}{t}\right)\left(1+o(1)\right)\right)\\
&= x^{-\alpha}\left(1+A(t)\frac{x^\beta-1}\beta\left(1+o(1)\right)
+ \frac{\alpha rG(y_p)}{1+rG(y_p)}
\frac{x^{-\alpha}-1}{-\alpha}\overline{F}\left(\hat{x}-\frac{1}{t}\right)\left(1+o(1)\right)
\right),
\end{align*}
which leads to the desired result, and this ends the proof.
\end{proof}
 
\end{lemma}

\begin{lemma}\label{lem:eighth}
 Let \((X,Y)\) be a bivariate risk vector satisfying the FGM dependence given by (\ref{eq:9}) with maginal distributions $F$ and $G$, respectively. Assume that \( F \in \mathrm{MDA}(\Lambda)\) with \( 0<\hat{x}\leq\infty \). 
 Then 
\[F_{{X|Y>y_p}} \in \mathrm{MDA}(\Lambda).\]


\begin{proof} Since \( F \in \mathrm{MDA}(\Lambda)\) with \( 0<\hat{x}\leq\infty \), there exists an auxiliary function $a(\cdot)$ such that (\ref{eq:6}) holds. Thus, for any $y\in\mathbb{R}$, by (\ref{eq:A1}) Taylor expansion,  as $x\uparrow\hat{x}$, it holds that
\begin{align*}
\frac{\overline{F}_{X|Y>y_p}\left(x+ya(x)\right)}
{\overline{F}_{X|Y>y_p}\left(x\right)}
&= \frac{\overline{F}\left(x+ya(x)\right)}
{\overline{F}\left(x\right)}\left[1-\frac{rG(y_p)}{1+rG(y_p)}\overline F(x+ya(x))\right]\left[1-\frac{rG(y_p)}{1+rG(y_p)}\overline F(x)\right]^{-1} \\
&=\frac{\overline{F}\left(x+ya(x)\right)}
{\overline{F}\left(x\right)}
\left(1-\frac{rG(y_p)}{1+rG(y_p)}\overline F(x+ya(x))\right)\\
&\quad\times
\left(1+\frac{rG(y_p)}{1+rG(y_p)}\overline F(x)(1+o(1))\right)\\
&=\frac{\overline{F}\left(x+ya(x)\right)}
{\overline{F}\left(x\right)}
\left(1-\frac{rG(y_p)}{1+rG(y_p)}
\left(\overline{F}\left(x+ya(x)\right)-\overline{F}\left(x\right)\right)\left(1+o(1)\right)\right)\\
&= \frac{\overline{F}\left(x+ya(x)\right)}
{\overline{F}\left(x\right)}
\left(1+ \frac{ rG(y_p)}{1+rG(y_p)}
(e^{-y}-1)\overline{F}\left(x\right)\left(1+o(1)\right)\right),
\end{align*}
which tends to $e^{-y}$ as $x\uparrow\hat{x}$. This leads to the desired result, and this ends the proof.
\end{proof}
\end{lemma}

\subsection{Proof of Theorem \ref{thm:main}}
\begin{proof} First, in view of $\overline F\in \mathrm{2RV}_{-\alpha,\rho}$, it follows from Lemma \ref{lem:fourth} that $\overline{F}_{X|Y>y_p}\in\mathrm{2RV}_{-\alpha,\rho}$ for $\rho=(-\alpha)\vee \beta$ with an auxiliary function $B(\cdot)$ as defined in Lemma \ref{lem:fourth}. For $F\in\mathrm{MDA}(\Phi_\alpha)$, the well-known fact that $\overline F\in\mathrm{2RV}_{-\alpha,\rho}$ is equivalent to $U\in\mathrm{2RV}_{1/\alpha,\rho/\alpha}$ implies that
$$U_{X|Y>y_p} \in 2\mathrm{RV}_{1/\alpha, \rho/\alpha}$$
with an auxiliary function $C(\cdot)=\alpha^{-2}B\left(U_{X|Y>y_p}(\cdot)\right)$.

Next, as stated in Introduction, for any risk variable $Z$, $\mathrm{HG}_q(Z)=\mathrm{HM}_{\tilde{q}}(Z)$ with $\tilde{q}=1-(1-q)^{1/k}$ for $k\geq 1$. Thus, applying the Fr\'echet case of Lemma \ref{lem:fifth} with $Z$ replaced by  the conditional function of $X$ conditioning on $Y>y_p$, by some simple computations, it yields that 
\begin{equation}\label{eq:A2}
\mathrm{CoHM}_{p,q}(X|Y) = c_1 F_{X|Y>y_p}^{\leftarrow}\bigl(1 - (1 - q)^k\bigr)\left[ 1 +  J_{\alpha, \rho, k}^{(1)} 
C\left((1-q)^{-k}\right)(1+o(1))\right] 
\end{equation}
with 
\begin{equation*}
C\left((1-q)^{-k}\right) 
= \frac{1}{\alpha^2} \left(A\left(U_{X|Y>y_p}\left(\frac{1}{(1-q)^k}\right)\right)+\frac{\alpha rG(y_p)} {1+rG(y_p)}\overline{F}\left(U_{X|Y>y_p}\left(\frac{1}{(1-q)^k}\right)\right)\right).\nonumber
\end{equation*}

Then, with the initial conclusion in hand, our next goal is to convert 
\(F_{X|Y>y_p}^{\leftarrow}\bigl(1 - (1 - q)^k\bigr)\) into the form of 
\(F^{\leftarrow}\left(1 - (1 - q)^k\right)\). To do so, we first let 
\begin{equation}
x_q=F^{\leftarrow}\left(1 - (1 - q)^k\right)=U\left(\frac{1}{(1-q)^k}\right).
\nonumber
\end{equation}
One can easily see that $x_q\rightarrow\infty$ as $q\uparrow 1$. By Lemma \ref{lem:first} we further deduce that
\begin{equation}\label{eq:A4}
{F}(x_q)=\left(1-(1-q)^k\right)\left(1+o(A(x_q)\right)=\left(1-(1-q)^k\right)+o(A(x_q)). 
\end{equation}
Combining Definition \ref{def:3} and (\ref{eq:A4}), we derive 
\begin{equation*}
\frac{U_{X|Y>y_p}\left(\frac{1}{(1-q)^k}\right)}{U_{X|Y>y_p}\left(\frac{1}{\overline{F}(x_q)}\right)}
=\left(1+o(A(x_q))\right)^{\frac{1}{\alpha}}\left[ 1+C\left(\frac{1}{\overline{F}(x_q)}\right)o(1)\right],
\end{equation*}
which implies that
\begin{eqnarray}\label{eq:relationship of U}
U_{X|Y>y_p}\left(\frac{1}{(1-q)^k}\right)=U_{X|Y>y_p}\left(\frac{1}{\overline{F}(x_q)}\right)\left(1+o(A(x_q))\right).
\end{eqnarray}
Next, by (\ref{eq:A1}), it follows from (\ref{eq:A4}) that 
\begin{equation}\label{eq:A5}
\frac{\overline{F}_{X|Y>y_p}(x_q)}{\overline{F}(x_q)}
=1+rG\left(y_p)\left[(1-(1-q)^k\right)+o(A(x_q)\right]=m_q+o(A(x_q)).
\end{equation}
Recalling that \(U_{X|Y>y_p} \in 2\mathrm{RV}_{1/\alpha, \rho/\alpha}\) with an auxiliary function $C(\cdot)=\alpha^{-2}B\left(U_{X|Y>y_p}(\cdot)\right))$, thus, by the uniform convergence in any compact subset of $\mathbb{R}^+$ in Definition \ref{def:3} and Lemma \ref{lem:first}, we obtain 
\begin{eqnarray*}
\frac{U_{X|Y>y_p}\left(\frac{1}{\overline{F}(x_q)}\right)}
{U_{X|Y>y_p}\left(\frac{1}{\overline{F}_{X|Y>y_p}(x_q)}\right)}-\left(\frac{\overline{F}_{X|Y>y_p}(x_q)}{\overline{F}(x_q)}\right)^{\frac{1}{\alpha}}=H_{\frac{1}{\alpha},\frac{\rho}{\alpha}}\left(m_1\right)\alpha^{-2}\left(A(x_q)+\frac{\alpha rG(y_p)}{1+rG(y_p)}\overline F(x_q)\right)(1+o(1)),
\end{eqnarray*}
where we used the fact $U_{X|Y>y_p}\left(\frac{1}{\overline{F}_{X|Y>y_p}}(x_q)\right)\sim x_q$  due to Lemma \ref{lem:first}. Setting $
 m_{q}=1+rG(y_p)\left(1-(1-q)^k\right)$ and substituting (\ref{eq:A5}) into the relation above yield that
\begin{eqnarray*}
&\quad&\frac{U_{X|Y>y_p}\left(\frac{1}{\overline{F}(x_q)}\right)}
{U_{X|Y>y_p}\left(\frac{1}{\overline{F}_{X|Y>y_p}(x_q)}\right)}\\
&=&m_q^{1/\alpha}\left(1+\frac{1}{\alpha}m_q^{-1}o(A(x_q))\right)+H_{\frac{1}{\alpha},\frac{\rho}{\alpha}}\left(m_1\right)\alpha^{-2}\left(A(x_q)+\frac{\alpha rG(y_p)}{1+rG(y_p)}\overline F(x_q)\right)(1+o(1))\\
&=&m_q^{1/\alpha}\left(1+\frac{m_1^{\rho/\alpha}-1}{\rho\alpha}\left(A(x_q)+\frac{\alpha rG(y_p)}{1+rG(y_p)}(1-q)^k\right)(1+o(1))\right),
\end{eqnarray*}
where in the first step we used the Taylor expansion. Therefore, observing that 
\begin{eqnarray*}
\frac{F_{X|Y>y_p}^{\leftarrow}\bigl(1 - (1 - q)^k\bigr)}
{F^{\leftarrow}\left(1 - (1 - q)^k\right)}=\frac{U_{X|Y>y_p}\left(\frac{1}{(1-q)^k}\right)}
{x_q}=\frac{U_{X|Y>y_p}\left(\frac{1}{(1-q)^k}\right)}{U_{X|Y>y_p}\left(\frac{1}{\overline{F}(x_q)}\right)}\cdot\frac{U_{X|Y>y_p}\left(\frac{1}{\overline{F}(x_q)}\right)}
{U_{X|Y>y_p}\left(\frac{1}{\overline{F}_{X|Y>y_p}(x_q)}\right)},
\end{eqnarray*}
we conclude that
\begin{eqnarray}\label{eq:A6}
&\quad&\frac{F_{X|Y>y_p}^{\leftarrow}\bigl(1 - (1 - q)^k\bigr)}
{F^{\leftarrow}\left(1 - (1 - q)^k\right)} \nonumber \\
&=&m_q^{1/\alpha}\left(1+\frac{m_1^{\rho/\alpha}-1}{\rho\alpha}\left(A(x_q)+\frac{\alpha rG(y_p)}{1+rG(y_p)}(1-q)^k\right)(1+o(1))\right).
\end{eqnarray}
The relation obtained in (\ref{eq:A6}) above also means that
\begin{equation}\nonumber
\frac{U_{X|Y>y_p}\left(\frac{1}{(1-q)^k}\right)}
{U\left(\frac{1}{(1-q)^k}\right)}\sim m_q^{\frac{1}{\alpha}}
\sim m_1^{\frac{1}{\alpha}}.
\end{equation}
This together with \(A(\cdot)\in \mathrm{RV}_\beta\) and \(\overline{F}\in \mathrm{RV}_{-\alpha}\) yields that
\[\frac{A\left(U_{X|Y>y_p}\left(\frac{1}{(1-q)^k}\right)\right)}
{A\left(U\left(\frac{1}{(1-q)^k}\right)\right)}\sim
\left(\frac{U_{X|Y>y_p}\left(\frac{1}{(1-q)^k}\right)}
{U\left(\frac{1}{(1-q)^k}\right)}\right)^\beta
\sim m_1^{\frac{\beta}{\alpha}}
\]
and
\[
\frac{\overline{F}\left(U_{X|Y>y_p}\left(\frac{1}{(1-q)^k}\right)\right)}
{\overline{F}\left(U\left(\frac{1}{(1-q)^k}\right)\right)}\sim
\left(\frac{U_{X|Y>y_p}\left(\frac{1}{(1-q)^k}\right)}
{U\left(\frac{1}{(1-q)^k}\right)}\right)^{-\alpha}
\sim m_1^{-1}.
\]
Thus, combining all of these results leads to
\begin{equation}\label{eq:A7}
C\left((1-q)^{-k}\right)\sim
\frac{1}{\alpha^2} \left(m_1^{\frac{\beta}{\alpha}} A\left(x_q\right)+
\frac{\alpha rG(y_p)} {1+rG(y_p)} m_1^{-1}
(1-q)^k\right).
\end{equation}
Finally, plugging (\ref{eq:A6}) and (\ref{eq:A7}) into (\ref{eq:A2}) and omitting higher-order terms yield the desired result as required and this ends the proof.
\end{proof}

\subsection{Proof of Theorem \ref{thm:second}}
\begin{proof}
First,  for convenience, set $h(t)=\overline F\left(\hat{x}-1/t\right)$. It is easy to verify that $$\left(\frac{1}{h}\right)^{\leftarrow}(t)=\frac{1}{\hat{x}-U\left(t\right)},$$ which leads to 
\begin{equation}\label{eq:Weibull+1}
F(F^{\leftarrow}(t))= t\left(1+o\left(A\left(1\Big/\left(\hat{x}-U\left(\frac{1}{1-t}\right)\right)\right)\right)\right)~~~ \mathrm{as} \quad t\uparrow 1,
\end{equation}
and
\begin{equation}\label{eq:Weibull+4}
\left(\hat{x}-U\left(1\Big/\overline F\left(\hat{x}-\frac{1}{t}\right)\right)\right)^{-1}=t\left(1+o(A(t))\right)~~~ \mathrm{as} \quad t\rightarrow \infty,
\end{equation}
due to Lemma \ref{lem:first}, where $A(t)$ is the auxiliary function of $\overline F\left(\hat{x}-1/t\right)$. 

Next, it follows from (\ref{eq:A1}) that
\begin{equation*}
\frac{\overline{F}_{X|Y>y_p}\left(\hat{x}-\frac{1}{t}\right)}
{\overline{F}\left(\hat{x}-\frac{1}{t}\right)}
= 1 +rG(y_p)F\left(\hat{x}-\frac{1}{t}\right).
\end{equation*}
Therefore, noticing that \( \overline{F} \left(\hat{x}-\frac{1}{\cdot}\right)\in 2\mathrm{RV}_{-\alpha,\beta}\) for some $\alpha>0$ and $\beta\leq 0$ with an auxiliary function  $A\left(\cdot\right)$, by Lemma \ref{lem:WEIBULL}, we can deduce that 
\begin{equation}\label{eq:Weibull+2}
\overline{F}_{X|Y>y_p}  \left(\hat{x}-\frac{1}{\cdot}\right)\in 2\mathrm{RV}_{-\alpha,\rho}
\end{equation}
for $\rho=(-\alpha)\vee \beta$ with an auxiliary function 
\begin{equation*}
B(\cdot)=A(\cdot)+\frac{\alpha rG(y_p)}{1+rG(y_p)}\overline{F}\left(\hat{x}-\frac{1}{\cdot}\right).
\end{equation*}.

Then, recall  the well-known fact that, for $F\in \mathrm{MDA}(\Psi_\alpha)$, $\overline F\left(\hat{x}-\frac{1}{\cdot}\right)\in \mathrm{2RV}_{-\alpha,\rho}$ with an auxiliary function $\tilde{A}(\cdot)$ is equivalent to $\hat{x}-U\in \mathrm{2RV}_{-1/\alpha,\rho/\alpha}$ with an auxiliary function $-\alpha^{-2}\tilde{A}\left(1/(\hat{x}-U(\cdot))\right)$. This result together with (\ref{eq:Weibull+2}) implies that
 \(\hat{x}-U_{X|Y>y_p} \in 2\mathrm{RV}_{-1/\alpha, \rho/\alpha}\) with an auxiliary function 
\begin{equation}\label{eq:Weibull+5}
C(\cdot)=-\alpha^{-2}B\left(1/(\hat{x}-U_{X|Y>y_p}(\cdot))\right).
\end{equation}
Similarly as above,  applying the Weibull case of Lemma \ref{lem:fifth}, by some simple computations, it yields that 
\begin{equation}\label{eq:A8}
\hat{x}-\mathrm{CoHM}_{p,q}(X|Y) = c_2 \left(\hat{x}-F_{X|Y>y_p}^{\leftarrow}\bigl(1 - (1 - q)^k\bigr)\right)\left[ 1 + J_{\alpha, \rho, k}^{(2)} C\left((1-q)^{-k}\right) (1+o(1))\right] 
\end{equation}
with  
\begin{eqnarray}\label{eq:Weibull+6}
C\left((1-q)^{-k}\right) 
&=& -\alpha^{-2} B\left(1/\left(\hat{x}-U_{X|Y>y_p}\left(\frac{1}{(1-q)^k}\right)\right)\right)\nonumber\\
&=& -\alpha^{-2} \left(A\left(1/\left(\hat{x}-U_{X|Y>y_p}\left(\frac{1}{(1-q)^k}\right)\right)\right)\right.\nonumber\\
&\quad&\left.+\frac{\alpha rG(y_p)}{1+rG(y_p)}\overline{F}\left(U_{X|Y>y_p}\left(\frac{1}{(1-q)^k}\right)\right)\right).
\end{eqnarray}

Finally, it suffices to convert 
\(F_{X|Y>y_p}^{\leftarrow}\bigl(1 - (1 - q)^k\bigr)\) into the form of 
\(F^{\leftarrow}\left(1 - (1 - q)^k\right)\) as before. By a treatment similar as in the proof of Theorem \ref{thm:main}, let $x_q=F^{\leftarrow}\left(1 - (1 - q)^k\right)=U\left((1-q)^{-k}\right)$, which implies that $x_q\uparrow\hat{x}$ as $q\uparrow 1$. It follows from (\ref{eq:A1}) and (\ref{eq:Weibull+1}) that
\begin{eqnarray}\label{eq:A9}
\frac{\overline{F}_{X|Y>y_p}(x_q)}{\overline{F}(x_q)}=1+rG(y_p)F(x_q)=m_q+o\left(A\left(\frac{1}{\hat{x}-x_q}\right)\right).
\end{eqnarray}
Recalling that \(\hat{x}-U_{X|Y>y_p} \in 2\mathrm{RV}_{-1/\alpha, \rho/\alpha}\) with an auxiliary function $C(\cdot)$ given by (\ref{eq:Weibull+5}), we conclude that
\begin{eqnarray*}
&\quad&\frac{\hat{x}-U_{X|Y>y_p}\left(\frac{1}{\overline{F}(x_q)}\right)}
{\hat{x}-U_{X|Y>y_p}\left(\frac{1}{\overline{F}_{X|Y>y_p}(x_q)}\right)}-\left(\frac{\overline{F}_{X|Y>y_p}(x_q)}{\overline{F}(x_q)}\right)^{-\frac{1}{\alpha}}\\
&=&H_{-\frac{1}{\alpha},\frac{\rho}{\alpha}}\left(m_1\right)(-\alpha^{-2})\left(A\left(z_q\right)+\frac{\alpha rG(y_p)}{1+rG(y_p)}\overline F\left(\hat{x}-\frac{1}{z_q}\right)\right)(1+o(1))
\end{eqnarray*}
with
\[z_q=\left(\hat{x}-U_{X|Y>y_p}\left(\frac{1}{\overline{F}_{X|Y>y_p}(x_q)}\right)\right)^{-1}.\]
Now, substituting (\ref{eq:A9}) into the relation above yields that 
\begin{eqnarray*}
&\quad&\frac{\hat{x}-U_{X|Y>y_p}\left(\frac{1}{\overline{F}(x_q)}\right)}
{\hat{x}-U_{X|Y>y_p}\left(\frac{1}{\overline{F}_{X|Y>y_p}(x_q)}\right)}
=m_q^{-1/\alpha}\left(1-\frac{1}{\alpha}m_q^{-1}o\left(A\left(\frac{1}{\hat{x}-x_q}\right)\right)\right)\\
&\quad&+H_{-\frac{1}{\alpha},\frac{\rho}{\alpha}}\left(m_1\right)(-\alpha^{-2})\left(A(z_q)+\frac{\alpha rG(y_p)}{1+rG(y_p)}\overline F\left(\hat{x}-\frac{1}{z_q}\right)\right)(1+o(1))\\
&=&m_q^{-1/\alpha}\left(1-\frac{m_1^{\rho/\alpha}-1}{\rho\alpha}\left(A\left(\frac{1}{\hat{x}-x_q}\right)+\frac{\alpha rG(y_p)}{1+rG(y_p)}(1-q)^k\right)(1+o(1))\right),
\end{eqnarray*}
where in the last step we used $z_q=(\hat{x}-x_q)^{-1}\left(1+o\left(A\left(\frac{1}{\hat{x}-x_q}\right)\right)\right)$  by (\ref{eq:Weibull+4}). 
Besides, following the same approach of (\ref{eq:relationship of U}), it is easy to verify 
\begin{equation*}
    \hat{x}-U_{X|Y>y_p}\left(\frac{1}{(1-q)^k}\right)=\hat{x}-U_{X|Y>y_p}\left(\frac{1}{\overline{F}(x_q)}\right)\left(1+o\left(A\left(\frac{1}{\hat{x}-x_q}\right)\right)\right).
\end{equation*}
Therefore, noticing that 
\begin{eqnarray*}
\frac{\hat{x}-F_{X|Y>y_p}^{\leftarrow}\bigl(1 - (1 - q)^k\bigr)}
{\hat{x}-F^{\leftarrow}\left(1 - (1 - q)^k\right)}
=\frac{\hat{x}-U_{X|Y>y_p}\left(\frac{1}{(1-q)^k}\right)}{\hat{x}-U_{X|Y>y_p}\left(\frac{1}{\overline{F}(x_q)}\right)}\cdot\frac{\hat{x}-U_{X|Y>y_p}\left(\frac{1}{\overline{F}(x_q)}\right)}
{\hat{x}-U_{X|Y>y_p}\left(\frac{1}{\overline{F}_{X|Y>y_p}(x_q)}\right)}\cdot\frac{z_q^{-1}}{\hat{x}-x_q},
\end{eqnarray*}
we conclude that
\begin{eqnarray}  \label{eq:A10}
&\quad&\frac{\hat{x}-F_{X|Y>y_p}^{\leftarrow}\bigl(1 - (1 - q)^k\bigr)}
{\hat{x}-F^{\leftarrow}\left(1 - (1 - q)^k\right)} \nonumber\\
&=&m_q^{-1/\alpha}\left(1-\frac{m_1^{\rho/\alpha}-1}{\rho\alpha}\left(A\left(\frac{1}{\hat{x}-x_q}\right)+\frac{\alpha rG(y_p)}{1+rG(y_p)}(1-q)^k\right)(1+o(1))\right). 
\end{eqnarray}
The relation obtained in (\ref{eq:A10}) above also means that
\begin{equation}\nonumber
\frac{\hat{x}-U_{X|Y>y_p}\left(\frac{1}{(1-q)^k}\right)}
{\hat{x}-U\left(\frac{1}{(1-q)^k}\right)}\sim m_q^{-\frac{1}{\alpha}}
\sim m_1^{-\frac{1}{\alpha}}.
\end{equation}
This together with \(A(\cdot)\in \mathrm{RV}_\beta\) and \(\overline{F}(\hat{x}-\frac{1}{\cdot})\in \mathrm{RV}_{-\alpha}\) yields that
\[
\frac{A\left(1/\left(\hat{x}-U_{X|Y>y_p}\left(\frac{1}{(1-q)^k}\right)\right)\right)}
{A\left(1/\left(\hat{x}-U\left(\frac{1}{(1-q)^k}\right)\right)\right)}
\sim\left(\frac{\hat{x}-U\left(\frac{1}{(1-q)^k}\right)}
{\hat{x}-U_{X|Y>y_p}\left(\frac{1}{(1-q)^k}\right)}\right)^\beta
\sim m_1^{\frac{\beta}{\alpha}}
\]
and
\[
\frac{\overline{F}\left(U_{X|Y>y_p}\left(\frac{1}{(1-q)^k}\right)\right)}
{\overline{F}\left(U\left(\frac{1}{(1-q)^k}\right)\right)}=\frac{\overline{F}\left(\hat{x}-\frac{1}{\hat{x}-U_{X|Y>y_p}\left(\frac{1}{(1-q)^k}\right)}\right)}
{\overline{F}\left(\hat{x}-\frac{1}{\hat{x}-U\left(\frac{1}{(1-q)^k}\right)}\right)}\sim
\left(\frac{\hat{x}-U\left(\frac{1}{(1-q)^k}\right)}
{\hat{x}-U_{X|Y>y_p}\left(\frac{1}{(1-q)^k}\right)}\right)^{-\alpha}
\sim m_1^{-1}.
\]
Thus, substituting all these results into (\ref{eq:Weibull+6}) leads to
\begin{equation}\label{eq:A11}
C\left((1-q)^{-k}\right)\sim
-\frac{1}{\alpha^2} \left(m_1^{\frac{\beta}{\alpha}} A\left(\frac{1}{\hat{x}-x_q}\right)+
\frac{\alpha rG(y_p)} {1+rG(y_p)} m_1^{-1}
(1-q)^k\right).
\end{equation}
Plugging (\ref{eq:A10})-(\ref{eq:A11}) into (\ref{eq:A8}) and omitting higher-order terms yield the desired result as required and this ends the proof.
\end{proof}

\subsection{Proof of Theorem \ref{thm:third}}

\begin{proof}
First, since $F\in \mathrm{MDA}(\Lambda)$ with an auxiliary function $a(\cdot)$, according to Lemma \ref{lem:eighth}, it holds that \(F_{X|Y>y_p} \in \mathrm{MDA}(\Lambda) \) with  the same auxiliary function $a(\cdot)$, which is equivalent to \(U_{X|Y>y_p} \in\mathrm{ERV}_{0} \) due to Theorem 1.1.6 in \citet{deHaan2006}.

Next, applying Lemma \ref{lem:fifth} directly gives 
\begin{equation}\label{eq:*1}
\mathrm{CoHM}_{p,q}(X|Y) = F_{X|Y>y_p}^{\leftarrow}\bigl(1 - (1 - q)^k\bigr) + 
a\left( \frac{1}{(1-q)^k} \right)
\lambda_k (1 + o(1))
\end{equation}
for \(\hat{x} = \infty\), and
\begin{equation}\label{eq:*2}
\hat{x}-\mathrm{CoHM}_{p,q}(X|Y) = \hat{x}- F_{X|Y>y_p}^{\leftarrow}
\bigl(1 - (1 - q)^k\bigr) - 
a\left( \frac{1}{(1-q)^k} \right)
\lambda_k (1 + o(1))
\end{equation}
for \(\hat{x} < \infty\).

Next, it suffices to convert 
\(F_{X|Y>y_p}^{\leftarrow}\bigl(1 - (1 - q)^k\bigr)\) into the form of 
\(F^{\leftarrow}\left(1 - (1 - q)^k\right)\) as before. To do so, we also let \(x_q=F^{\leftarrow}\left(1 - (1 - q)^k\right)=U\left((1-q)^{-k}\right)\).  It follows from Lemma \ref{lem:first+0} that \(\overline{F}(x_q)\sim(1-q)^k\) as \(q\uparrow1\), which together with Lemma \ref{lem:first+1} further implies 
\begin{eqnarray}\label{eq:Gumbel+0}
U\left((1-q)^{-k}\right)=U\left(\frac{1}{\overline F(x_q)}\pm 0\right)=U\left(\frac{1}{\overline F(x_q)}\right)+o\left(a\left(\frac{1}{\overline F(x_q)}\right)\right). \nonumber
\end{eqnarray}
Similarly, one can also derive
\begin{eqnarray}\label{eq:Gumbel+1}
U_{X|Y>y_p}\left((1-q)^{-k}\right)=U_{X|Y>y_p}\left(\frac{1}{\overline F(x_q)}\right)+o\left(a\left(\frac{1}{\overline F(x_q)}\right)\right). \nonumber
\end{eqnarray}
Moreover, it follows from (\ref{eq:A1}) that  
\begin{eqnarray}\label{eq:Gumbel+2}
\frac{\overline{F}_{X|Y>y_p}(x_q)}{\overline{F}(x_q)}=1+rG(y_p)F(x_q)=m_q+o\left(1\right).   \nonumber
\end{eqnarray}
Now we first deal with the case of $\hat{x}=\infty$. Note that
\begin{flalign}\label{eq:Gumbel+3}
&F_{X|Y>y_p}^{\leftarrow}\bigl(1 - (1 - q)^k\bigr)-
F^{\leftarrow}\left(1 - (1 - q)^k\right)      \nonumber\\
=&U_{X|Y>y_p}\left(\frac{1}{\overline{F}(x_q)}\right)-U_{X|Y>y_p}\left(\frac{1}{\overline{F}_{X|Y>y_p}(x_q)}\right)\nonumber\\
\quad&+U_{X|Y>y_p}\left(\frac{1}{\overline{F}_{X|Y>y_p}(x_q)}\right)-U\left(\frac{1}{\overline{F}(x_q)}\right)+o\left(a\left(\frac{1}{\overline F(x_q)}\right)\right).\nonumber
\end{flalign}
Clearly, it follows from the definition of ERV that
\begin{equation}\label{eq:Gumbel+4}
U_{X|Y>y_p}\left(\frac{1}{\overline{F}(x_q)}\right)-
U_{X|Y>y_p}\left(\frac{1}{\overline{F}_{X|Y>y_p}(x_q)}\right)= a\left(\frac{1}{\overline{F}_{X|Y>y_p}(x_q)}\right)
\log\left(m_1\right)\left(1+o(1)\right).     \nonumber
\end{equation}
The replacement $t$ by $1/\overline F(t)$ in Lemma \ref{lem:first+1} also yields 
\begin{equation}\label{eq:Gumbel+5}
 U\left(\frac{1}{\overline F(x_q)}\right)=x_q+o\left(a\left(\frac{1}{\overline F(x_q)}\right)\right).  \nonumber
\end{equation}
Similarly, we also have
\begin{equation}\label{eq:Gumbel+6}
U_{X|Y>y_p}\left(\frac{1}{\overline F_{X|Y>y_p}(x_q)}\right)=x_q+o\left(a\left(\frac{1}{\overline F_{X|Y>y_p}(x_q)}\right)\right). \nonumber
\end{equation}
Combining these results yields that
\begin{flalign}\label{eq:Gumbel+7}
F_{X|Y>y_p}^{\leftarrow}\bigl(1 - (1 - q)^k\bigr)=
F^{\leftarrow}\left(1 - (1 - q)^k\right)+a\left(\frac{1}{(1-q)^k}\right)
\log\left(m_1\right)\left(1+o(1)\right),
\end{flalign}
where we used 
\begin{flalign}
a\left(\frac{1}{\overline F_{X|Y>y_p}(x_q)}\right)\sim a\left(\frac{1}{\overline F(x_q)}\right)\sim a\left(\frac{1}{(1-q)^k}\right)\nonumber
\end{flalign}
due to $a(\cdot)\in\mathrm{RV}_0$. This leads to the desired result.\\
For the case of $\hat{x}<\infty$, by an analogous treatment, observe that
\begin{flalign}\label{eq:Gumbel+9}
&\left(\hat{x}-F_{X|Y>y_p}^{\leftarrow}\bigl(1 - (1 - q)^k\bigr)\right)-
\left(\hat{x}-F^{\leftarrow}\bigl(1 - (1 - q)^k\bigr)\right)      \nonumber\\
=&\left(\hat{x}-U_{X|Y>y_p}\left(\frac{1}{\overline{F}\left(\hat{x}-1/(\hat{x}-x_q)^{-1} \right)}\right)\right)-
\left(\hat{x}-U_{X|Y>y_p}\left(\frac{1}{\overline{F}_{X|Y>y_p}\left(\hat{x}-1/(\hat{x}-x_q)^{-1} \right)}\right)\right)\nonumber\\
\quad&+\left(\hat{x}-U_{X|Y>y_p}\left(\frac{1}{\overline{F}_{X|Y>y_p}\left(\hat{x}-1/(\hat{x}-x_q)^{-1} \right)}\right)\right)-
\left(\hat{x}-U\left(\frac{1}{\overline{F}\left(\hat{x}-1/(\hat{x}-x_q)^{-1} \right)}\right)\right) \nonumber \\
\quad&+o\left(a\left(\frac{1}{\overline F\left(\hat{x}-1/(\hat{x}-x_q)^{-1} \right)}\right)\right) \nonumber \\
=&\left(\hat{x}-U_{X|Y>y_p}\left(\frac{1}{\overline{F}(x_q)}\right)\right)-
\left(\hat{x}-U_{X|Y>y_p}\left(\frac{1}{\overline{F}_{X|Y>y_p}(x_q)}\right)\right)\nonumber\\
\quad&+\left(\hat{x}-U_{X|Y>y_p}\left(\frac{1}{\overline{F}_{X|Y>y_p}(x_q)}\right)\right)-
\left(\hat{x}-U\left(\frac{1}{\overline{F}(x_q)}\right)\right)+o\left(a\left(\frac{1}{\overline F(x_q)}\right)\right).
\nonumber
\end{flalign}
From the definition of ERV we obtain that
\begin{flalign*}
&\left(\hat{x}-U_{X|Y>y_p}\left(\frac{1}{\overline{F}(x_q)}\right)\right)-
\left(\hat{x}-U_{X|Y>y_p}\left(\frac{1}{\overline{F}_{X|Y>y_p}(x_q)}\right)\right)  \nonumber\\
=& a\left(\frac{1}{\overline{F}_{X|Y>y_p}(x_q)}\right)
\log\left(m_1\right)\left(1+o(1)\right).     \nonumber
\end{flalign*}
The replacement $t$ by $1/\overline F(t)$ or $1/\overline F\left(\hat{x}-1/(\hat{x}-t)^{-1}\right)$ in Lemma \ref{lem:first+1} yields 
\begin{equation}\label{eq:Gumbel+11}
 \hat{x}-U\left(\frac{1}{\overline F(x_q)}\right)=\hat{x}-x_q+o\left(a\left(\frac{1}{\overline F(x_q)}\right)\right), \nonumber
\end{equation}
and
\begin{equation}\label{eq:Gumbel+12}
\hat{x}-U_{X|Y>y_p}\left(\frac{1}{\overline F(x_q)}\right)=\hat{x}-x_q+o\left(a\left(\frac{1}{\overline F_{X|Y>y_p}(x_q)}\right)\right). \nonumber
\end{equation}
Combining these results yields that
\begin{equation}\label{eq:*4}
\hat{x}-F_{X|Y>y_p}^{\leftarrow}\bigl(1 - (1 - q)^k\bigr)
= \hat{x}-F^{\leftarrow}\left(1 - (1 - q)^k\right)
-a \left(F^{\leftarrow}\left(1 - (1 - q)^k\right)\right)\log(m_1)\left(1+o(1)\right).
\end{equation}
Finally, plugging (\ref{eq:Gumbel+7}) and (\ref{eq:*4}) into (\ref{eq:*1})-(\ref{eq:*2}) yields the desired result as required, and this ends the proof.
\end{proof}

\end{document}